\documentclass[leqno]{amsart}
\usepackage{amsmath}
\usepackage{amsfonts,amssymb}

\usepackage{graphicx,color,epsfig}
\usepackage{hyperref}

\theoremstyle{plain}
\newtheorem{theorem}{Theorem}[section]
\newtheorem{definition}[theorem]{Definition}

\newtheorem{lemma}[theorem]{Lemma}

\newtheorem{proposition}[theorem]{Proposition}

\theoremstyle{remark}
\newtheorem{remark}[theorem]{Remark}

\usepackage{graphicx,subfigure}
\usepackage{pdfsync}

\def\C{{\mathbb C}}
\def\Z{{\mathbb Z}}
\def\N{{\mathbb N}}
\def\R{{\mathbb R}}

\def\CC{{\mathcal C}}
\def\LL{{\mathcal L}}

\def\e{{\rm e}}
\def\tr{{\rm tr}\,}

\def\Id{{\rm Id}}

\def\eps{\varepsilon}
\def\sgn{{\rm sgn}}
\def\op_#1{\mathrel{\mathop{{\rm op}_{#1}}}}

\def\build#1_#2^#3{\mathrel{
\mathop{\kern 0pt#1}\limits_{#2}^{#3}}}

\def\td_#1,#2{\mathrel{
\mathop{\build\longrightarrow_{#1\rightarrow #2}^{}}}}

\def\lim_#1,#2{\mathrel{
\mathop{\build{\rm lim}_{#1\rightarrow#2}^{}}}}

\def\limsup_#1,#2{\mathrel{
\mathop{\build{\rm limsup}_{#1\rightarrow#2}^{}}}}

\def\liminf_#1,#2{\mathrel{
\mathop{\build{\rm liminf}_{#1\rightarrow#2}^{}}}}

\def\aref#1{(\ref{#1})}
\def\eps{\varepsilon}

\def\1{{\bf 1}}
\def\0{{\bf 0}}

\def\op{{\rm op}}
\def\Ran{{\rm Ran}}
\def\supp{{\rm supp}}
\def\ol{\overline}
\def\d{{\rm d}}

\title{An Egorov Theorem for  avoided crossings of eigenvalue surfaces}

\author[C. Fermanian]{Clotilde Fermanian Kammerer}
\address{Universit\'e Paris Est, UMR 8050 du CNRS, France}
\email{clotilde.fermanian@univ-paris12.fr}

\author[C. Lasser]{Caroline Lasser}
\thanks{C. Lasser acknowledges support by the German Research Foundation
(DFG), Collaborative Research Center SFB-TR 109.}
\address{Zentrum Mathematik, Technische Universit\"at M\"unchen, 80290 M\"unchen, Germany}
\email{classer@ma.tum.de}

\date{\today}


\begin{document}

\begin{abstract}
We study nuclear propagation through avoided crossings of
electron energy levels.
We construct a surface hopping semigroup, which gives an Egorov-type description of the dynamics. The underlying time-dependent
Schr\"odinger equation  has a two-by-two
matrix-valued potential, whose eigenvalue surfaces have an
avoided crossing. Using microlocal
normal forms reminiscent of the Landau-Zener problem, we prove
convergence to the true solution in the semi-classical limit.  
\end{abstract}

\maketitle


\section{Introduction}\label{sec:intro}

We consider a system of Schr\"odinger equations 
\begin{equation}\label{eq:schro}
 \left\{\begin{array}{l}
i\eps\partial_t\psi^\eps_t = -\frac{\eps^2}{2} \Delta_q\psi^\eps_t +V(q)\psi^\eps_t,\;\;(t,q)\in\R\times\R^d\\
\psi^\eps_{t=0}=\psi^\eps_0,
\end{array}
\right.
\end{equation}
where $V$ is a smooth function on $\R^d$, whose values are $2\times2$ real symmetric matrices, 
$$
V(q) = \alpha(q){\rm Id} + \begin{pmatrix}\beta(q)&\gamma(q)\\ \gamma(q) & -\beta(q)\end{pmatrix},\qquad q\in\R^d.
$$
The smooth functions $\alpha,\beta,\gamma\in{\mathcal C}^\infty(\R^d,\R)$ are of subquadratric growth such that the Schr\"odinger operator in equation (\ref{eq:schro}) is essentially self-adjoint, and there exists a unique solution for all times $t\in\R$.  We aim at the semi-classical limit $\eps\to0$ in situations, where the two eigenvalues $\lambda^+(q)\ge \lambda^-(q)$ of the potential matrix~$V(q)$ get close to each other and  non-adiabatic transitions occur to leading order in $\eps$. Our method of proof allows for rather general initial data $(\psi^\eps_0)_{\eps>0}$ that are uniformly bounded in $L^2(\R^d,\C^2)$.

\medskip
Schr\"odinger systems with matrix-valued potentials can be rigorously derived in the context of the Born--Oppenheimer approximation of molecular quantum dynamics, and we refer to~\cite{ST} and~\cite{MS} where this theory is carefully carried out.  Born--Oppenheimer theory also applies for the present two-level system, provided that the eigenvalues $\lambda^+$ and $\lambda^-$ are uniformly separated, 
that is, if there exists a small gap parameter $\delta_0>0$, independent of the semi-classical parameter $\eps>0$, such that the gap function
\begin{equation}\label{def:g}
g(q)= \lambda^+(q)-\lambda^-(q)
\end{equation}
satisfies $g(q)\ge \delta_0$ for all $q\in\R^d$. In this situation, the eigenspaces are adiabatically decoupled in the following sense: If $\Pi^\pm(q)$ denote the eigenprojectors onto the eigenspaces of $V(q)$, then for initial data with  $\psi^\eps_0=\Pi^+\psi^\eps_0$, one obtains only a small non-adiabatic contribution $\|\Pi^-\psi^\eps(t)\|_{L^2} = O(\eps)$ at time $t$, and the analogous statement holds true if $\psi^\eps_0=\Pi^-\psi^\eps_0$. If the gap-condition is violated because the gap becomes small (with respect to $\eps$) or vanishes, then adiabatic decoupling no longer holds. 
For concrete molecular systems, the semi-classical parameter $\eps$ and the gap parameter $\delta_0$ are given numbers, and an asymptotic analysis just taking into account the smallness of the semi-classical parameter $\eps$ does not provide the necessary information.

\medskip
Especially for Schr\"odinger systems, non-adiabatic transitions have been of interest over decades, 
since they occur in many applications ranging from atmospheric chemistry to photochemistry, see 
the recent perspective article \cite{Tu2}. Typically, non-adiabatic phenomena are attributed to avoided or conical crossings of eigenvalues. Conical crossings occur, when the eigenvalue gap vanishes and the eigenprojectors have a conical singularity at these points. Generic conical crossings have been classified by symmetry, see \cite{Hag1}, and have been analysed by semiclassical wavepackets \cite{Hag1} as well as by pseudodifferential operators \cite{CdV1,FG02,LT}, 
see also the kinetic model for graphene \cite{FM} that encorporates a conical crossing. Here we aim at the analysis of avoided crossings. For them,  the following definition has been proposed \cite{Hag2}:

\begin{definition}\label{def:avoided_hag} 
Suppose $V(q,\delta)$ is a family of real symmetric $2\times 2$ matrices depending smoothly on $q\in\Omega$ and 
$\delta\in I$, $V\in\mathcal C^\infty(\Omega\times I,\R^{2\times 2})$, where $\Omega$ is an open subset of $\R^d$ and $I\subset\R$ some interval containing $0$. Suppose that $V(q,\delta)$ has two eigenvalues $\lambda^+(q,\delta)$ and $\lambda^-(q,\delta)$ that depend continuously on $q$ and $\delta$. Assume that 
$$
\{q\in\Omega,\;\; \lambda^+(q,0)=\lambda^-(q,0)\}
$$ 
is a non-empty submanifold of $\Omega$, such  that $\lambda^+(q,\delta)\neq\lambda^-(q,\delta)$ for all $q\in\Omega$ and $\delta\neq0$. Then we say that $V(q,\delta)$ has an avoided crossing of eigenvalues. 
\end{definition}  

Avoided crossings have a similar symmetry classification \cite{Hag2} as the conical intersections. 
In \cite{HJ1,HJ2,Rou}, the perturbation parameter $\delta$ of Definition~\ref{def:avoided_hag} has been linked with the semiclassical parameter $\eps$, and a leading order analysis of 
$$
i\eps\partial_t\psi^\eps_t = -\tfrac{\eps^2}{2}\Delta_q \psi^\eps_t + V(q,\sqrt\eps)\psi^\eps_t,\qquad 
\psi^\eps_{t=0} = \psi^\eps_0
$$
has been carried out for families of initial data $(\psi^\eps_0)_{\eps>0}$ which are semiclassical wavepackets. 
For Schr\"odinger systems in one space dimension, avoided crossings have also been considered without assuming that the perturbation $\delta$ and the semiclassical parameter $\eps$ are coupled. In this situation, the non-adiabatic contributions are exponentially small with respect to $\eps$. In \cite{HJ3}, it is assumed that the potential $V(q)$ of the Schr\"odinger system~(\ref{eq:schro}) belongs to a family $V(q,\delta)$ with an avoided crossing such that $V(q)=V(q,\delta_0)$ for all $q\in\R$ and some fixed $\delta_0>0$. Then, the scattering wave function of semiclassical wavepackets is determined together with its non-adiabatic contributions. In \cite{BGT,BG}, superadiabatic representations of one-dimensional avoided crossings have been developed together with an explicit heuristic formula for the outgoing nonadiabatic component. Our results interpolate between the existing ones in the following sense. On the one hand, 
we allow for general families $(\psi^\eps_0)_{\eps>0}$ of initial data in $L^2(\R^d,\C^2)$ without restricting to coherent states or a single space dimension. Also, we will not explicitly link the semiclassical parameter $\eps$ and the gap parameter $\delta$. On the other hand, we will only reach for the leading order behaviour with respect to $\eps$.


\subsection{Wigner transforms}  It is impossible to directly study the densities 
$$
n_\pm^\eps(q,t)= | \Pi^\pm(q)\psi^\eps(q,t) |_{\C^2}^2
$$ 
or the  dynamics of the so-called level populations
$$
t\mapsto \int_{\R^d} n_\pm^\eps(t,q) dq
$$ 
for general initial data. Thus, we focus on providing an asymptotic description for the time evolution of the
Wigner transform of~$\psi^\eps(q,t)$ in a suitable
$\eps$-dependent scaling,
$$W^\eps\!\left(\psi^\eps_t\right)\!(q,p)=
(2\pi)^{-d}\int_{\R^d} \psi^\eps\!\left(q-\tfrac{\eps}{2}v,t\right)\otimes \ol \psi^\eps\!\left(q+\tfrac{\eps}{2}v,t\right)\,{\rm e}^{i\,v\cdot p}\,\d v$$
with $(q,p)\in\R^{2d}$. The Wigner transform plays
the role of a generalized probability density on phase space. For
square integrable wave functions $\psi\in L^2(\R^d,\C^2)$, the Wigner function
$W^\eps(\psi)$ is a square integrable function on phase space with
values in the space of Hermitian $2\times 2$ matrices. One recovers for
example the position density by
$$n_\pm^\eps(q,t)=
{\rm tr}\int_{\R^d}
\Pi^\pm(q)W^\eps(\psi^\eps_t)(q,p)\,\d p.
$$
Besides, the action of the Wigner function against compactly
supported smooth test functions
$a\in\CC^\infty_c(\R^{2d},\C^{2\times 2})$ is simply expressed in
terms of the semi-classical pseudodifferential operator of symbol
$a$, which is defined by
$$
\op_\eps(a)\psi(q)=(2\pi\eps)^{-d} \int_{\R^{2d}}
a\left(\tfrac12(q+q'),p\right){\rm
e}^{\tfrac{i}{\eps}p\cdot(q-q')}\psi(q')\,\d q'\,\d q
$$
for $\psi\in L^2(\R^d,\C^2)$. Indeed, we have
$$\int_{\R^{2d}}\,{\rm tr}\left(W^\eps(\psi)(q,p)a(q,p)\right)\,\d
q\,\d
p=\left(\op_\eps(a)\psi\,,\,\psi\right)_{L^2(\R^d,\C^{2})}.
$$
The Wigner transform is perfectly suited for the analysis of quadratic functions of the wave function, which do not require all the phase information.

Our aim is the study of the diagonal parts of the Wigner transform
\begin{eqnarray*}
\Pi^\pm(q) W^\eps(\psi^\eps_t)\Pi^\pm(q)=w^\pm_\eps(t)\Pi^\pm(q),\\
w_\pm^\eps(t)={\rm tr} \left(\Pi^\pm(q) W^\eps(\psi^\eps_t)\Pi^\pm(q)\right),
\end{eqnarray*}
and to describe the evolution of the coefficients $w_\pm^\eps(t)$ in terms of $w_+^\eps(0)$  and $w_-^\eps(0)$ as $\eps\to0$. The oscillatory dynamics of the off-diagonal part of the Wigner function implies that it can be neglected far from the crossing set (see Remark~\ref{rem:appendix} in the Appendix). However, these effects could  restrict our results, see the comments after our main Theorem~\ref{theorem} and the corresponding numerical experiment in \S\ref{sec:dual}.


\subsection{Egorov's theorem}

We consider the classical flow
$$
\Phi^t_\pm: \R^{2d}\to\R^{2d}\,,\quad
\Phi^t_\pm(q_0,p_0)=\left(q^\pm(t),p^\pm(t)\right)
$$
associated with the Hamiltonian curves of $\Lambda^\pm(q,p)=\tfrac{|p|^2}{2}+\lambda^\pm(q)$. These curves   are solutions to the
Hamiltonian systems
\begin{equation}\label{eq:clastraj}
\left\{\begin{array}{l} \dot q^\pm(t)=p^\pm(t),\;\;\dot
p^\pm(t)=-\nabla \lambda^\pm\left(q^\pm(t)\right),\\
q^\pm(0)=q_0,\;\;p^\pm(0)=p_0 \end{array}\right.
\end{equation}
which can be solved for all $t\in\R$, since the maps $q\mapsto \lambda^\pm(q)$ are smooth for eigenvalues, which do not intersect each other.

\medskip
If the eigenvalues are uniformly separated from each other, 
then the classical flows~$\Phi^t_\pm$ are enough for an
approximate description of the dynamics up to an error of order
$\eps$. Indeed, 
the action of the diagonal part of the Wigner transform on scalar test functions~$a\in{\mathcal C}^\infty_c(\R^{2d+1},\C)$ obeys
\begin{equation}
\label{eq:ct} \int_{\R^{2d+1}}
\left(w^\eps_\pm(t)- w^\eps_\pm(0)
\circ\Phi_\pm^{-t}\right)\!(q,p)\,a(t,q,p)\,\d(t,q,p) = O(\eps).
\end{equation}
Such dynamical descriptions in the spirit of Egorov's theorem are
well established, see for example \cite{GMMP}. 

\medskip
If the gap $g(q)=\lambda^+(q)-\lambda^-(q)$ is not uniformly bounded from below and small, but not too small, this description is still valid. More precisely, one proves in \cite{FL08} (see also the proof in Appendix~A) that as long as the trajectories of $\Phi^t_\pm$ which reach the support of the observable $a$ remain in a zone where  $g(q)>R\sqrt\eps = \eps^{3/8}$ for $R=R(\eps) = \eps^{-1/8}$, then
\begin{equation*}
 \int_{\R^{2d+1}}
\left(w_\pm^\eps(t)-
w_\pm^\eps(0)\circ\Phi_\pm^{-t}\right)\!(q,p)\,a(t,q,p)\,\d(t,q,p)=O(\eps^{1/8}) ,
\end{equation*}
where the error estimate just depends on derivatives of the potential matrix $V$ and the symbol $a$, while it is independent of the gap parameter $\delta_0$.

\medskip
However, on regions with smaller eigenvalue gap the approximation of the diagonal Wigner
components $w_\pm^\eps(t)$ by mere classical transport is no longer valid, and there are
non-adiabatic transitions between the levels.  
The components propagated until the  crossing region on one level may pass (partially or
utterly) the other level.

\subsection{Surface hopping}
For a particular isotropic conical crossing \cite{LT} and later for general conical crossings \cite{FL08}, it has been proved
 that the diagonal parts of the Wigner transform can effectively be described by the following random walk construction: We consider a classical trajectory $\Phi_+^t(q,p)$ with associated weight $w_+^\eps(q,p,t)$. If the gap function 
$$
t\mapsto g(q^+(t))
$$ 
attains a local minimum at time $t^*$ for the phase space point $(q^*,p^*)$, such that $g(q^*)\le R\sqrt\eps$, then one opens a new trajectory $\Phi_-^{t-t^*}(q^*, p^*_{out})$ with
$$
 p^*_{out}=p ^*+\omega^*,\;\;\omega^*=g(q^*)\tfrac{p^*}{|p^*|^2}.
$$
The two trajectories $\Phi_+^{t-t^*}(q^*, p^*)$ and   $\Phi_-^{t-t^*}(q^*, p^*_{out})$ are equipped  with the weights
\begin{eqnarray*}
w^+_\eps(q^*,p^*, (t^*)^+) & = & \left(1-T_{\eps}(q^*,p^*)\right)w^+_\eps(q^*,p^*,(t^*)^-),\\
w^-_\eps(q^*,p^*_{out}, (t^*)^+)& = & T_{\eps}(q^*,p^*)w^+_\eps(q^*,p^*,(t^*)^-),
\end{eqnarray*}
respectively. The transition probability is given by the Landau--Zener formula
\begin{equation}\label{eq:transfert}
T_{\eps}(q^*,p^*)=\exp\!\left(-\frac{\pi}{4\, \eps} \frac{g(q^*)^2}{|{\rm det}(p^*\cdot\nabla_qV_0(q^*))|^{1/2}}\right),
\end{equation}
where $V_0(q^*)$ denotes the trace-free part of the potential matrix $V(q^*)$. The analogous construction applies to the classical trajectories entering the region of small gap on the other eigenvalue surface.

\medskip
This combination of classical transport and Landau--Zener transitions yields an easy algorithm for the numerical simulation of non-adiabatic quantum dynamics, see \cite{FL12} and its applications to a three-dimensional model of the pyrazine molecule~\cite{LS} and the twelve-dimensional ammonia cation \cite{BDLT}. Its striking properties are, that only classical trajectories, local gap minima along classical trajectories and the Landau--Zener formula~\aref{eq:transfert} are required. 
Many other surface hopping
algorithms exist in the chemical literature starting with the pioneering work of Tully and Preston~\cite{TP}, and it is worth mentioning that they are equally applied for systems with avoided or conical eigenvalue crossings. In high dimensions, surface hopping algorithms are computationally much less demanding than the discretization of the full wave function and thus often a popular choice. Despite the intense research activity in chemical physics on these algorithms, 
there are very few mathematical results on their justification.

\subsection{Aim and organisation of the paper.}
We are interested in extending the Landau--Zener random walk  through conical crossings~\cite{LT,FL08} to the case of avoided crossings, thus obtaining a unified treatment for conical intersections and avoided crossings, regardless of the respective sizes of the gap and the semi-classical parameter. As far as we know, 
this unified treatment of both crossing types and its rigorous mathematical analysis is new, see also \cite{FL12} for comments on the subject. Following~\cite{Hag2},  we assume that the potential matrix $V$ presents an 
avoided crossing in the sense of the definition below.
 
\begin{definition}\label{def:avoided} 
 Let the potential
 \[
 V(q) = \alpha(q)\Id + \begin{pmatrix}\beta(q) & \gamma(q)\\ \gamma(q) & -\beta(q)\end{pmatrix},\qquad q\in\R^d,
 \]
 be a smooth function on $\R^d$, whose values are $2\times 2$ real symmetric matrices. Denote by $\lambda^\pm(q)$ and $g(q) = \lambda^+(q)-\lambda^-(q)$ the eigenvalues and the gap function of $V(q)$ and by $V_0(q)$ its trace-free part, that is, 
\begin{equation}\label{def:V0}
V_0(q) = \begin{pmatrix}\beta(q) & \gamma(q)\\ \gamma(q) & -\beta(q)\end{pmatrix},\qquad q\in\R^d.
\end{equation} 
We say that $V$ has an {\em avoided gap} in an open subset $\Omega\subset\R^d$ if it satisfies the following conditions :
\begin{enumerate}
\item
There exists $\delta_0>0$, which is the minimum of $g$ in $\Omega$, 
\item Let 
$$
S_0=\{q\in\R^d,\;\; g(q)=\delta_0\}. 
$$
The set $S_0\cap\Omega$ is a hypersurface.
\item 
There exists a system of local coordinates $(y_1,y')$ with $S_0\cap\Omega=\{y_1=0\}\cap\Omega$ and
$$V_0(y)=y_1V_1(y)+V_2(y')$$
where $V_1(y)$ is an invertible matrix, while $V_1(y)$ and $V_2(y)$ are linearly independent for all $y\in\Omega$.
\end{enumerate}
 \end{definition}
 
An illustrative easy example for an avoided crossing in the sense of Definition~\ref{def:avoided} is provided by
\begin{equation}\label{eq:simple_example}
V(q) = \begin{pmatrix}q_1 & \delta_0\\ \delta_0 & -q_1\end{pmatrix},\qquad q\in\R^d,\quad \delta_0>0.
\end{equation}
Here, the eigenvalues 
\[
\lambda^\pm(q) = \pm\sqrt{q_1^2+\delta_0^2}
\] 
have a global minimal gap of size $\delta_0$ at the hyperplane $S_0 = \{q\in\R^d, q_1= 0\}$, and the potential matrix can be written as $V(q) = q_1 V_1 + V_2$ with
\[
V_1 = \begin{pmatrix}1 & 0 \\ 0 & -1\end{pmatrix},\qquad V_2 = \begin{pmatrix}0 & \delta_0\\ \delta_0 & 0\end{pmatrix}.
\] 
On the contrary, the matrix 
$$V(q) = \sqrt {q_1^2+\delta_0^2}\,  \begin{pmatrix}1 & 0 \\ 0 & -1\end{pmatrix}$$
not satisfy the assumptions of Definition~\ref{def:avoided}.

\begin{remark}
We note that if one can write the minimal gap set in local coordinates as $S_0\cap\Omega = \{y_1=0\}$ and $V_0(y) = y_1 V_1(y) + V_2(y')$ with $V_1(y)$ an invertible matrix, 
then $V_1(y)$ and $V_2(y')$ are necessarily linearly independent, in the sense, that there exists no smooth function $f:\Omega\to\R$ with $V_1(y) = f(y) V_2(y')$ for all $y\in\Omega$.
\end{remark}

\begin{remark}
The avoided crossing of Definition~\ref{def:avoided} is associated with a minimal gap manifold $S_0\subset\R^d$ of codimension one. 
In the symmetry classification of avoided crossings given in \cite{Hag2}, higher codimensions also occur. 
We expect that our analysis of the codimension one case can be generalized, see also Remark~\ref{rem:generalize}.
\end{remark}

If potential $V$ has an 
avoided gap with minimal gap size~$\delta_0>0$ (in the sense of Definition~\ref{def:avoided}), then one can construct  a family $V(q,\delta)$ with an avoided crossing of eigenvalues in the sense of Definition~\ref{def:avoided_hag} with two crucial properties. First, we have
$$
V(q) = V(q,\delta_0),\qquad q\in\Omega,
$$ 
and second the mapping $(q,\delta)\mapsto V(q,\delta)$ has a conical intersection of its eigenvalues for $(q,\delta)\in S_0\times\{\delta=0\}$, 
see Theorem~\ref{prop:parametrization} below. Considering momenta $p\in\R^d$ which are transverse to the hypersurface~$S_0$  ensures that the crossing is generic in $T^*(\R^d_q\times\R_\delta)$ in the sense of \cite{CdV1} and \cite{FG03}. This link between avoided and conical crossings -- which is already indicated in Colin de Verdi\`ere's paper \cite[\S2.8]{CdV1} -- allows us to redevelop the proof strategy of \cite{FL08}.
A crucial step in our new proof is an elementary normal form construction inspired by \cite{FG02} that is explicit enough to keep $\delta$ as a controlled parameter.

\medskip
We start in Section~\ref{sec:result} by providing the precise mathematical statement of our result. In 
Section~\ref{sec:strategy} we prove the relation between avoided and conical crossings. In Section~\ref{sec:reduction} 
we perform an elementary reduction to a Landau--Zener model parametrized by the gap parameter $\delta$. Both these new results are 
crucial for the proof of the surface hopping approximation in Section~\ref{sec:mainproof}. We then describe the 
associated surface hopping algorithm in Section~\ref{sec:numerics} and present numerical experiments for Tully's well-known avoided crossing models~\cite{Tu1}. 
The Appendix~\ref{proof:prop} presents the proof of classical transport in the zone of large gap.


\section{An Theorem}\label{sec:result}

We now give a precise statement for the Egorov type description of the propagation of the diagonal part of the Wigner transform for systems with an avoided eigenvalue crossing in the sense of Definition~\ref{def:avoided}. We consider the classical trajectories $(q^\pm(t),p^\pm(t))$ of the Hamiltonian systems \aref{eq:clastraj} and monitor 
the phase space points, where the classical trajectories attain a {\em local
minimal gap between the two eigenvalues.} At such points we have
\begin{equation}\label{jumpcond}
\frac{\d}{\d t}\left(g(q^\pm(t))
\right)=p^\pm(t)\cdot \nabla_q g(q^\pm(t))=0.
\end{equation}
Therefore, one performs an effective
non-adiabatic transfer of weight, whenever a trajectory passes the
set
$$
\Sigma_{\eps} = \left\{(q,p)\in\R^{2d}\mid 
\;g(q)\le R\sqrt\eps
 ,\;\;p\cdot\nabla g(q)=0\right\}
$$
where $R=R(\eps) = \eps^{-1/8}\gg1$. This choice of $R$ is motivated from the analysis in~\cite{FL08} and ensures that the largest occuring  error terms 
$R^3\sqrt\eps$ and $R^{-5}\eps^{-1/2}$ are of the same order $\eta_\eps=\eps^{1/8}$.

\subsection{The random trajectories}
We attach the labels~$j=-1$ and~$j=+1$ to the phase space $\R^{2d}$
and consider the random  trajectories
$$
{\mathcal T}_{\eps}^{(q,p,j)}:[0,+\infty)\rightarrow
\R^{2d}\times\{-1,+1\},
$$
with
$$
{\mathcal T}_{\eps}^{(q,p,j)}(t)=\left(\Phi^t_j(q,p),j\right)\quad\mbox{as long as} \quad\Phi^t_{j}(q,p)\not\in \Sigma_{\eps}.$$ 
Whenever the deterministic flow $\Phi^t_j(q,p)$ hits the set
$\Sigma_{\eps}$ at a point $(q^*,p^*)$, a jump 
$$
(q^*,p^*,j) \to (q^*,p^*+j\,\omega^*,-j)
$$ 
occurs with the transition probability $T_{\eps}(q^*,p^*)$  defined in~\aref{eq:transfert}. The drift
 \begin{equation}\label{def:drift}
 \omega^* = \omega^*(q^*,p^*) = \frac{g(q^*)}{|p^*|^2}\, p^*,
 \end{equation}
is applied to preserve the energy of the trajectories
$$
\Lambda^\pm(q,p) = \tfrac12|p|^2 +  \lambda^\pm(q) = \tfrac12|p|^2 + \alpha(q) \pm\tfrac12 g(q)
$$
up to order $R^2\eps$.
Indeed, let us suppose that the incoming trajectory is on the plus mode. Then, one chooses the momentum $p^*_{out}=p^*+ \omega^*$ of the trajectory generated on the minus mode such that its energy $\Lambda^-(q^*,p^*_{out})$ satisfies
$$\Lambda^-(q^*,p^*_{out})=\Lambda^+(q^*,p^*)+O(R^2\eps).$$ Since 
$\Lambda^-(q^*,p^*_{out})=
 \tfrac{1}{2} |p^*+ \omega^*|^2 +\alpha(q^*)-\tfrac{1}{2}g(q^*)$, 
it is enough to choose $\omega^*$ such that 
$ \omega^*\cdot p^* = g(q^*),
$
whence~\aref{def:drift}. 

\begin{remark}\label{rem:energy}
Let us comment about various aspects of the drift. The drift $p^*\pm\omega^*$ will be crucial later on for localizing the solution at a distance of size~$R\sqrt\eps$ to the energy surfaces $\{\tau+\Lambda^\pm(q,p)=0\}$. 
The drift is performed in the momentum coordinates, since the difference of the two Hamiltonian vector fields 
\[
H_{\lambda^\pm}(q,p) = (p, -\nabla\lambda^\pm(q))
\] 
vanishes identically in the position coordinates. Moreover, after a change of space-time coordinates, see Section~\ref{subsubsec:drift}, the drift is exact when performed in this direction. A similar drift has been used in \cite{HJ1} and~\cite{HJ2} for analysing wave packet propagation through avoided crossings.

The transversality condition that will be stated in assumption~(A0) of our main Theorem~\ref{theorem} excludes trajectories 
with small momenta at the jump manifold. Therefore, 
$\omega^*$ is of the order of the gap size and thus bounded by $R\sqrt\eps$.
Outside the jump manifold $\Sigma_\eps$ the gap and thus the drift are large. However, the Landau--Zener transition rates are exponentially small there. 
Consequently, in this regime the drift would be harmless (if performed).
\end{remark}

\subsection{The semigroup} Within a
bounded time-interval $[0,T]$, each path
$$
(q,p,j)\to {\mathcal T}_{\eps}^{(q,p,j)}(t)
$$
only has a finite number of jumps, remains in bounded regions of
$\R^{2d}\times\{-1,+1\}$, and is smooth away from the jump
manifold $\Sigma_{\eps}\times\{-1,+1\}$. Hence, the random
trajectories define a Markov process
$$
\left\{ X^{(q,p,j)} \mid (q,p,j)\in\R^{2d}\times\{-1,+1\}\right\}.
$$
The associated transition function $P(p,q,j;t,\Gamma)$ describes
the probability of being at time $t$ in the set $\Gamma\subset
\R^{2d}\times \{-1,+1\}$ having started in $(q,p,j)$. Its action
on the set 
$$
{\mathcal B} = \left\{ f:\R^{2d}\times\{-1,+1\}\to\C \mid f \;\mbox{is measurable, bounded}\right\}
$$ 
defines a semigroup
$(\LL_{\eps}^t)_{t\ge0}$ by
$$
{\mathcal L}_{\eps}^t\,f(q,p,j) =
\int_{\R^{2d}\times\{-1,+1\}} f(x,\xi,k)\,P(q,p,j;t,\d(x,\xi,k)).
$$

\begin{remark}
We associate with $f\in{\mathcal B}$ two functions $f_\pm:\R^{2d}\to\C$ via
\begin{equation}\label{fpm}
f_\pm(q,p)=f(q,p,\pm 1).
\end{equation}
Reversely,  relation~\aref{fpm} implies that two bounded measurable functions $f_\pm$ on $\R^{2d}$ define a function $f\in{\mathcal B}$. We shall use this identification all over the paper. 
\end{remark}

We now define the action of the semigroup on Wigner functions
by duality. More precisely, let $\psi\in L^2(\R^d,\C^2)$ and $W^\eps(\psi)$ be its Wigner transform.  Denote by
$$ 
w^\eps_\pm(\psi)(q,p) =  {\rm tr }\left(\Pi^\pm(q)W^\eps(\psi)(q,p)\right)
$$ 
the diagonal components of $W^\eps(\psi)$ and define $w^\eps(\psi)\in{\mathcal B}$ according to relation \aref{fpm}.
For $a\in{\mathcal B}$ such that $a_+$ and $a_-$ have compact support, we set
$$
\left( w^\eps(\psi),a\right) =\int_{\R^{2d}}w^\eps_+(\psi)(q,p) \,a_+(q,p)\, \d(q,p) + \int_{\R^{2d}} w^\eps_-(\psi)(q,p)\, a_-(q,p)\, \d(q,p)
$$
and define ${\mathcal L}_{\eps}^t w^\eps(\psi)\in{\mathcal B}$ by
$$
\left({\mathcal L}_{\eps}^t w^\eps(\psi), a\right) =\left(w^\eps(\psi),{\mathcal L}_{\eps}^t a\right).
$$

\subsection{The result}\label{subsec:result}
Let $V$ be a potential matrix presenting an avoided crossing of eigenvalues in the sense of Definition~\ref{def:avoided}, the notations of which we shall use in the following. 
The semi-group $({\mathcal L^t_\eps})_{t>0}$ approximates the non-adiabatic dynamics generated by this avoided crossing in $\Omega\subset\R^d$, if we assume the following: 

\subsubsection{Initial data (A0)}
The initial data $(\psi^\eps_0)_{\eps>0}$ is a bounded family in
$L^2\!\left(\R^d,\C^2\right)$ associated either with $\Ran\Pi^+$
or $\Ran\Pi^-$, meaning that either
$$
\left|\left|\Pi^-\psi^\eps_0\right|\right|_{L^2(\R^d,\C^2)}=
O(\eps^{\beta_1}),\qquad \beta_1\ge 1/32,
$$
or the analogous condition on $\Pi^+\psi^\eps_0$ holds. We suppose
that the initial data are localized away from $S_0\cap\Omega$, 
 that is, there is some $C_0>0$ such that
$$
\int_{\{\d(q,S_0)<C_0\}} | W^\eps(\psi^\eps_0)(q,p)|\, \d q\, \d p =
O(\eps^{\beta_2}),\qquad \beta_2\ge1/32.
$$
The classical trajectories issued from the support of $W^\eps(\psi^\eps_0)$ reach their minimal gap points in $S_0\cap\Omega$. Moreover, 
the initial data is localized away from the set 
$$
\left\{(q_0,p_0)\in\R^{2d}\mid \exists t>0: q^\pm(t)\in S_0,
\;  p^\pm(t) \in T_{q^\pm(t)}S_0\right\},
$$
which contains those points, which are transported to the minimal gap manifold $S_0$ and have gained momenta which are not transverse to $S_0$.

\subsubsection{Observables (A1)}
The observable $a\in{\mathcal B}$ satisfies $a_\pm\in{\mathcal C}_c^\infty\left(\R^{2d},\C\right)$ and has its support at a distance larger than $\eps^{\beta_3}$
from $S_0$, i.\ e.\
$$
\d({\rm supp}_{(q,p)}(a_\pm),S_0)\gg \eps^{\beta_3},\qquad\beta_3\ge 1/32.
$$

\subsubsection{Time-interval (A2)}
Let $T>0$. Within the time-interval $[0,T]$, the classical trajectories issued from the support of $W^\eps(\psi^\eps_0)$ reach their minimal gap points only once.

\medskip
These assumptions on the initial data, the observables, and the time interval allow us to effectively describe the dynamics through an avoided crossing by surface hopping.

\begin{theorem}\label{theorem}
Let $\eps>0$ and $\psi^\eps$ be the solution of the Schr\"odinger equation
$$
i\eps\partial_t\psi^\eps_t = -\tfrac{\eps^2}{2} \Delta_q\psi^\eps_t +V(q)\psi^\eps_t,\qquad \psi^\eps_{t=0}=\psi^\eps_0,
$$
where the potential $V$ has an 
avoided crossing in the sense of Definition~\ref{def:avoided} with a 
gap parameter $\delta_0\in ]0,1]$.
Assuming (A0), (A1) and (A2), we have for all test functions $\chi\in{\mathcal C}_c^\infty([0,T])$ a constant $C>0$
\begin{equation}
\label{eq:approx}
\left| \int_0^T \chi(t) \left(w^\eps(\psi^\eps_t)-{\mathcal L^t_\eps w^\eps(\psi^\eps_0)},a\right) \d t \right| \le C\,\eps^{1/32},
\end{equation}
where the constant $C$ depends on a finite number of upper bounds of derivatives of the smooth functions $\alpha,\beta,\gamma$ defining the potential $V$ and $a,\chi$ and of lower bounds of the determinant of the matrix $V_1$.
\end{theorem}

\medskip
The semigroup $(\mathcal{L}^t_\eps)_{t\ge0}$ crucially depends on the jump manifold $\Sigma_\eps$, that comprises those points in phase space with $g(q)\le R\sqrt\eps$, 
$R=R(\eps) = \eps^{-1/8}$, where the classical trajectories $(q^\pm(t),p^\pm(t))$ attain locally minimal surface gaps.  If the minimal gap size $\delta_0>0$ of the avoided crossing is larger than $R\sqrt\eps$, then the jump manifold is the empty set, so that Theorem~\ref{theorem} reduces to a leading order description of expectation values for block-diagonal observables by mere classical transport, see Appendix~\ref{proof:prop}. As a consequence, $\delta_0\le R\sqrt\eps$ is the only  interesting regime, and the key issue is to prove the hopping formula locally, close to any point of  $S_0\cap\Omega$.
This is done by
using the possibility to parametrically link the avoided crossing with a conical intersection. This construction is carried out in Section~\ref{sec:strategy},  where we prove that close to any point of $\Omega$, the potential $V$ is embeddable in a parametrized family of potentials. We then prove a result  analogous to Theorem~\ref{theorem} for the family of solutions to the Schr\"odinger equation associated with the parametrized potentials (see Theorem~\ref{theorem:delta-result} below). The different steps of the proof of the surface hopping approximation are then developed in Section~\ref{sec:mainproof}, via a reduction to a Landau--Zener model performed in Section~\ref{sec:reduction}. 

\medskip
An interesting feature of Theorem~\ref{theorem} is that it justifies using a surface algorithm without assessing the size of the gap with respect to~$\varepsilon$. This is of major interest for applications, since for ``real'' molecular quantum systems the explicit comparison of $\delta_0$ and $\varepsilon$ might be difficult. The algorithm takes into account the three main regimes:
\begin{enumerate}
\item $\delta_0\gg\sqrt\eps$ propagation along the eigenvalue surfaces
\item $\delta_0\sim \sqrt\eps$ partial transition between eigenspaces
\item $\delta_0\ll\sqrt\eps$ total transition between the eigenspaces
\end{enumerate}
In the context of semi-classical wave packet propagation, \cite{HJ1} and \cite{HJ2} have analysed the second regime, while \cite{Rou} has considered the first and third one.
In particular, if there are several sizes of minimal gaps in different open subsets, Theorem~\ref{theorem} proves that the algorithm can be used and the transition process will be automatically adapted to the gap size. The Born--Oppenheimer result \eqref{eq:ct} uses plain classical transport and has an error constant 
that tends to infinity for shrinking minimal gap size $\delta_0$. In constrast, the error bound of Theorem~\ref{theorem} only depends on bounds of the potential $V$ that can be controlled with respect to $\delta_0$. (The lower bound on the determinant of $V_1$ is not related to the gap size.)

\medskip
We note that the transversality condition of assumption (A0) is crucial for the microlocal normal form we use for effectively describing the 
nonadiabatic transitions. For the simple example \eqref{eq:simple_example}, it means that the set $\{(q,p)\in\R^{2d}, q_1=p_1=0\}$ is negligible for the 
trajectories of
\[
\dot q = p,\qquad \dot p = \mp \frac{q_1}{\sqrt{q_1^2+\delta_0^2}} (1,0,\ldots,0) 
\]
that are issued from the support of the initial Wigner function $W^\eps(\psi^\eps_0)$.  The first condition of assumption (A0) can be relaxed to initial data associated with 
both $\Ran\, \Pi^+$ and $\Ran\,\Pi^-$, 
provided that the trajectories for both modes do not arrive simultaneously at the same phase space point of the  jump manifold $\Sigma_\eps$. However, as illustrated by the numerical experiments for the dual avoided crossing in Section~\ref{sec:dual}, simultaneous arrival at the jump manifold is the situation where the off-diagonal components of the Wigner transform become relevant such that the present surface hopping approximation breaks down.

\begin{remark}\label{rem:generalize}
The avoided crossing of Definition~\ref{def:avoided} has a minimal gap manifold $S_0\subset\R^d$ of codimension one.
The results of \cite{FG03,Fe06,FL08} on eigenvalue crossings of codimension three and five allow to extend Theorem~\ref{theorem} to avoided crossings with minimal gap manifolds of higher codimension as well. 
\end{remark}

\section{Reduction to a conical intersection}\label{sec:strategy}

We now introduce a family of potentials $(V_0(q,\delta))_{\delta\in I}$ locally extending the trace-free part $V_0(q)$ of our original potential. We verify 
that $V_0(q,\delta)$ viewed as a function on $\R^d\times I$ has a generic codimension~two crossing for $q\in S_0$ and $\delta=0$ in the sense of ~\cite{CdV1} and~\cite{FG03}, respectively. Then, we describe the parametrized Schr\"odinger system that we shall consider afterwards.

\subsection{Parametrization of the gap}

We start by constructing the family of trace-free potentials $(V_0(q,\delta))_{\delta\in I}$ that locally extends the original potential $V_0(q)$ by adding the gap size as an additional coordinate.

\begin{theorem}\label{prop:parametrization} 
Let $V$ have an
avoided crossing in $\Omega\subset\R^d$, in the sense of Definition~\ref{def:avoided}, with minimal gap hypersurface $S_0$.  
Then there  exists an open interval $I\subseteq\R$ with $0,\delta_0\in I$, an open subset $\widetilde\Omega\subseteq\Omega$ and two functions $\beta(q,\delta)$ and  $\gamma(q,\delta)$ smooth on $\widetilde\Omega\times I$ and affine in $\delta$,  such that the matrix
$$V_0(q,\delta)=\begin{pmatrix} \beta(q,\delta) & \gamma(q,\delta) \\ \gamma(q,\delta) & -\beta(q,\delta)\end{pmatrix}$$
satisfies the following properties:
\begin{enumerate}
\item We have $V_0(q,\delta_0)=V_0(q)$ for all $q\in\widetilde\Omega$.
\item The eigenvalue gap
$$
g(q,\delta)=2\sqrt{\beta(q,\delta)^2+\gamma(q,\delta)^2}
$$ 
of $V_0(q,\delta)$ is of minimal size $|\delta|$ on $S_0$, that is, 
$$
g(q,\delta)\ge|\delta|\;\mbox{for all}\;q\in\widetilde\Omega,\quad
g(q,\delta)=|\delta|\;\mbox{if and only if}\;q\in S_0.
$$
\item There exists a smooth orthogonal matrix $R(q)$ and two smooth functions $\widetilde \beta(q,\delta)$  and $\widetilde \gamma(q)$, where $\widetilde \beta$ is affine in $\delta$, such that 
$$R(q)V_0(q,\delta)R(q)^*=\begin{pmatrix}
\widetilde\beta(q,\delta) & \delta \widetilde\gamma(q) \\
\delta \widetilde\gamma(q)  & - \widetilde\beta(q,\delta)
\end{pmatrix}$$
and $\widetilde\gamma(q)\neq 0$ for all $q\in\widetilde\Omega$. All derivatives of $\widetilde\beta$ are bounded.
\item
If $y=(y_1,y')$ are local coordinates such that $S_0\cap\widetilde\Omega = \{y_1=0\}\cap\widetilde\Omega$, then 
\[
\forall k\in\N \,\exists c_k>0 : \sup_{y\in\widetilde\Omega} |\partial^k_{y_1}\widetilde\gamma(y)| < c_k
\]
and
\[
\forall \alpha\in\N^{d-1} \,\exists c_\alpha>0 : \sup_{y\in\widetilde\Omega} |\partial^\alpha_{y'}\widetilde\gamma(y)| < c_\alpha {|y_1|\over \delta_0},
\] 
while all other derivatives of the function $\widetilde\gamma$ are of the order $1/\delta_0$. The derivative bounds  involve a lower bound on the determinant of the matrix~$V_1(y)$ 
of the decomposition $V_0(y) = y_1 V_1(y) + V_2(y')$.
\end{enumerate} 
\end{theorem}

\begin{proof}
We work close to some $q_0\in S_0 $ in  local coordinates $y=(y_1,y')$ such that $S_0\cap\Omega=\{y_1=0\}\cap\Omega$ and 
$$
V_0(y)=y_1V_1(y)+V_2(y')\;\;
\text{with}\;\;V_1(y)\;\text{invertible}
$$ 
for all $y\in\Omega$.
Setting $y_1=0$, one obtains that $V_2$ and consequently $V_1$ are trace-free on $\Omega$. We denote 
$$
V_1(y)=\begin{pmatrix} a_1(y)& b_1(y)\\ b_1(y) & -a_1(y)\end{pmatrix},\quad
V_2(y')=\begin{pmatrix} a_2(y')& b_2(y')\\ b_2(y') & -a_2(y')\end{pmatrix}
$$
and write the gap as
\begin{align*}
g(y)^2 &= 4y_1^2 \left(a_1(y)^2 + b_1(y)^2\right) + 8y_1\left(a_1(y)a_2(y') + b_1(y)b_2(y')\right)\\
& \quad + 4\left(a_2(y')^2 + b_2(y')^2\right) 
\end{align*}
for all $y\in\Omega$. From the relation $g(y)^2=\delta_0^2$ for all $y\in S_0$ we then deduce 
\[
\delta_0^2 = 4(a_2(y')^2+b_2(y')^2)\;\;\mbox{for all}\;\; y\in\Omega.
\]
We define for $\delta\in\R$
$$
V_0(y,\delta) = y_1 V_1(y)+{\frac{\delta}{\delta_0}} V_2(y')
$$
such that for all $y\in\Omega$
$$
V_0(y,\delta_0)=V_0(y),\qquad g(y,\delta_0)=g(y)
$$
and
\begin{eqnarray*}
g(y,\delta)^2 = 4y_1^2\left(a_1(y)^2 + b_1(y)^2\right)+ 8y_1 \frac{\delta}{\delta_0}
\left(a_1(y)a_2(y')+b_1(y)b_2(y')\right) + \delta^2.
\end{eqnarray*} 
With respect to the original coordinates, this construction means
$$
V_0(q,\delta)= \begin{pmatrix}\beta(q,\delta) & \gamma(q,\delta)\\ \gamma(q,\delta) & -\beta(q,\delta)\end{pmatrix}
$$
with
$$
\beta(q,\delta)=y_1a_1(y)+\frac{\delta}{\delta_0} a_2(y'),\quad
\gamma(q,\delta)=y_1b_1(y)+\frac{\delta}{\delta_0} b_2(y').
$$

Let us prove now that the gap $g(y,\delta)$ is minimal on $S_0$ for $\delta$ in some open interval~$I$ which contains $[0,\delta_0]$.
Since the gap $g(y)$ is minimal on $S_0\cap\Omega$, we have $\nabla_y (g(y)^2)=0$ for all $y\in S_0\cap\Omega$. Consequently, 
\[
a_1(0,y')a_2(y')+b_1(0,y')b_2(y') = 0,\qquad y=(0,y')\in S_0\cap\Omega.
\]
Therefore, there exist an open subset $\widetilde\Omega\subseteq\Omega$ and a continuous function $\Gamma:\widetilde\Omega\to\R$ such that 
$$a_1(y)a_2(y')+b_1(y)b_2(y')=y_1\Gamma(y)$$
 and 
\[
g(y,\delta)^2=4y_1^2\left(a_1(y)^2+b_1(y)^2+2\frac{\delta}{\delta_0}\Gamma(y)\right) + \delta^2.
\]
Now it remains to find an open interval $I$ such that
\begin{equation}\label{relation}
a_1(y)^2+b_1(y)^2+2\frac{\delta}{\delta_0}\Gamma(y) >0
\end{equation}
for all $(y,\delta)\in\widetilde\Omega\times I$ with $y_1\neq0$. We observe that 
$$
g(y,\delta_0)^2=g(y)^2>\delta_0^2\;\;{\rm for}\;\; y_1\not=0
$$
implies
$$
4y_1^2\left(a_1(y)^2+b_1(y)^2+2\Gamma(y)\right)>0\;\;{\rm for}\;\; y_1\not=0,
$$
while the invertibility of $V_1(y)$, $y\in\widetilde\Omega$, implies
$$
a_1(y)^2+b_1(y)^2 > 0,\qquad y\in\widetilde\Omega.
$$
Therefore, the affine function
$$
\delta\mapsto a_1(y)^2 + b_1(y)^2 + 2\frac{\delta}{\delta_0} \Gamma(y) 
$$
takes nonnegative values in $\delta=0$ and $\delta=\delta_0$ and thus for any $\delta\in]0,\delta_0[$, which yields \eqref{relation}. 
We note that the we can choose the interval $I$ small enough such that the quotient $\delta/\delta_0$ remains bounded for all $\delta\in I$. 
Consequently, the functions $\beta(\cdot,\delta)$ and $\gamma(\cdot,\delta)$ have smooth bounded derivatives.

\medskip
For the rotation of $V_0(q,\delta)$ we set
$$
A(q)= - {b_1(y(q))\over\sqrt{b_1(y(q))^2+a_1(y(q))^2}},\;\; B(q)={a_1(y(q))\over\sqrt{b_1(y(q))^2+a_1(y(q))^2}}
$$
and note that $A$ and $B$ are smooth functions with bounded derivatives, where the bound involves a lower bound on the determinant of $V_1$.
We also define the smooth rotation matrix $R(q) $ of angle $\theta(q)$ such that 
$$
{\rm cos} \left(2\theta(q)\right)=B(q),\qquad{\rm sin} \left(2\theta(q)\right)=A(q).
$$
Then, we have 
\begin{eqnarray*}
R(q) V_0(q,\delta) R(q)^* & = & \begin{pmatrix}
{\rm cos} \,\theta & -{\rm sin} \,\theta\\
{\rm sin} \,\theta & {\rm cos} \,\theta 
\end{pmatrix}
 \begin{pmatrix}\beta & \gamma\\ \gamma & -\beta\end{pmatrix}
 \begin{pmatrix}
{\rm cos} \,\theta & {\rm sin} \,\theta\\
-{\rm sin} \,\theta & {\rm cos} \,\theta 
\end{pmatrix}\\
& = &  \begin{pmatrix}
\beta\,{\rm cos} (2\theta)-\gamma\,{\rm sin}(2\theta) & \beta\,{\rm sin} (2\theta)+\gamma\,{\rm cos}(2\theta)\\
\beta\,{\rm sin} (2\theta)+\gamma\,{\rm cos}(2\theta) & -\beta\,{\rm cos} (2\theta)+\gamma\,{\rm sin}(2\theta) 
\end{pmatrix}\\
& = & 
\begin{pmatrix}
\widetilde\beta(q,\delta) & \delta \widetilde\gamma(q) \\
\delta \widetilde\gamma(q)  & - \widetilde\beta(q,\delta)
\end{pmatrix},
\end{eqnarray*}
where we define the $\delta$-affine function 
\begin{eqnarray*}
&&\widetilde \beta(q,\delta) = B(q)\beta(q,\delta) -A(q)\gamma(q,\delta)\\
&&= y_1(q)\sqrt{a_1(y(q))^2+b_1(y(q))^2} + \frac{\delta}{\delta_0} {a_2(y'(q))a_1(y(q))+b_1(y(q))b_2(y'(q))\over \sqrt {b_1(y(q))^2+a_1(y(q))^2}},
\end{eqnarray*}
whose derivatives are bounded functions, since $\delta/\delta_0$ is uniformly bounded. The smooth function $\widetilde\gamma(q)$ is defined by
\begin{eqnarray*}
\widetilde \gamma(q) & = & {1\over \delta} \left(A(q)\beta(q,\delta)+B(q)\gamma(q,\delta)\right)\\
& = & \frac{-b_1(y(q)) a_2(y'(q))+a_1(y(q)) b_2(y'(q))}{\delta_0 \sqrt{b_1(y(q))^2+a_1(y(q))^2}}.
\end{eqnarray*}
Observing that $a_2$ and $b_2$ only depend on $y'$ and that 
\[
a_2(y')^2+b_2(y')^2 = \tfrac{1}{4} \delta_0^2 \le\delta_0^2,\qquad y=(y_1,y')\in\widetilde\Omega,
\] 
we deduce that 
$\partial_{y_1}^k \widetilde\gamma$ is a smooth bounded function for any $k\in\N$. Using that $a_1a_2+b_1b_2=0$ on $S_0$, we observe 
\[
\widetilde\beta(q,\delta) = 0,\qquad q\in S_0\cap\widetilde\Omega,
\]
and $\delta^2 = \beta(\cdot,\delta)^2 + \gamma(\cdot,\delta)^2 = 
\widetilde\beta(\cdot,\delta)^2 + \delta^2\,\widetilde\gamma^2$,  
and we deduce that 
$\widetilde \gamma^2 = 1$ on $S_0$. Since $\widetilde\gamma$ is non-vanishing, 
due to the linear independence of $V_1(y)$ and $V_2(y')$, we then conclude that 
$\partial_{y'}^\alpha\widetilde\gamma(0,y')=0$ for all $\alpha\in\N^{d-1}$.
A Taylor expansion together with the rough estimate, 
that derivatives of $\widetilde\gamma$ are of the order $1/\delta_0$, yields the claimed bound on 
$\partial_{y'}^\alpha \widetilde\gamma(y)$ for $y\in\widetilde\Omega$.
\end{proof}

\subsection{The geometry of the crossing}
We add half the trace to the $\delta$-parametrized trace-free family of Theorem~\ref{prop:parametrization} and consider
\begin{equation}\label{def:V}
V(q,\delta) = \alpha(q)\Id + V_0(q,\delta),\qquad (q,\delta)\in\widetilde\Omega\times I,
\end{equation}
with $\alpha(q) = \frac12{\rm tr}V(q)$, such that the original potential can be written as
\[
V(q) = V(q,\delta_0),\qquad q\in\widetilde\Omega.
\]
We now verify that the symbol of the corresponding time-dependent Schr\"odinger operator, the matrix-valued function 
$$
P(q,p,\tau,\delta):=\left(\tau+\tfrac{|p|^2}{2}\right){\rm Id} +V(q,\delta),
$$
has a generic codimension two crossing on $S_0\times\{\delta=0\}$ in the sense of~\cite[\S1]{CdV1} and \cite[\S1]{FG03}. It has to satisfy the following two properties: 
\begin{enumerate}
\item The gap function satisfies
$$
g(q,\delta)=0\quad\text{if and only if}\quad (q,\delta)\in (S_0\cap\widetilde\Omega)\times\{\delta=0\},
$$
according to point (2) of Theorem~\ref{prop:parametrization}. 
\item The Poisson bracket 
$$
\left\{\tau+\tfrac{|p|^2}{2}+\alpha(q)\;,\; V_0(q,\delta)\right\}=p\cdot\nabla_q V_0(q,\delta)
$$
is invertible for $(q,\delta)\in (S_0\cap\widetilde\Omega)\times\{\delta=0\}$  and those momenta $p\in\R^d$ which are transverse to~$S_0$ at $q$.  This is implied by the following
Lemma~\ref{lem:transverse}.
\end{enumerate}

\begin{lemma}\label{lem:transverse} 
Let $V$ have an 
 avoided crossing in the sense of Definition~\ref{def:avoided} with minimal gap hypersurface $S_0$.  
Let $q_0\in S_0$ and $p_0\in\R^d$ transverse to $S_0$ at $q$. Then,  there exists a neighborhood $\Omega_1\subset\R^{2d}$ of $(q_0,p_0)$, independent of $\delta$ such that for all $\delta\in I$ and $(q,p)\in\Omega_1$,
$$
p\cdot \nabla_q V_0(q,\delta)\;\text{is invertible and}\;\; p\cdot\nabla_q\widetilde\beta (q,\delta)<0,
$$
where $(V_0(q,\delta))_{\delta\in I}$ and $\widetilde\beta(q,\delta)$ are defined as in Theorem~\ref{prop:parametrization}.
\end{lemma}

\begin{proof}
Let $\widetilde\Omega$ be the open set of Theorem~\ref{prop:parametrization}.
We again work close to some point $q_0\in S_0\cap\widetilde\Omega$ in  local coordinates $y=(y_1,y')$ such that $S_0\cap\widetilde\Omega=\{y_1=0\}\cap\widetilde\Omega$ and 
$$
\widetilde \beta(q,\delta) = \frac{y_1}{\sqrt{a_1(y)^2 + b_1(y)^2}} \left( a_1(y)^2 + b_1(y)^2 + \frac{\delta}{\delta_0} \Gamma(y)\right),
$$
where by~(\ref{relation}),
$$
a_1(y)^2 + b_1(y)^2 + 2\frac{\delta}{\delta_0} \Gamma(y)\ge 0,\qquad y\in\widetilde\Omega,\quad \delta\in I.
$$
Then, for all $(q,p)$,
$$
p\cdot \nabla_q \widetilde\beta(q,\delta) = \frac{p\cdot\nabla_q  y_1}{\sqrt{a_1(y)^2 + b_1(y)^2}} \left( a_1(y)^2 + b_1(y)^2 + \frac{\delta}{\delta_0} \Gamma(y)\right)+ y_1p\cdot \gamma(y,\delta)
$$
where the function $p\cdot \gamma(y,\delta)$ is bounded for $\delta \in I$ and $(p,q)$ in any bounded set. 
Since $y_1=0$ is an equation of the hypersurface $S_0$ in $\widetilde\Omega$ and $p_0$ is transverse to $S_0$ at $q_0$, we have $p_0\cdot\nabla_q y_1(q_0)\not=0$. Therefore, if necessary, we turn $y_1$ into $-y_1$, so that 
$$
p_0\cdot\nabla_q \widetilde \beta(q_0,\delta)<0,
$$
for all $\delta\in I$. Besides, by setting a bound on $y_1$ and $p$, we can find a $\delta$-independent neighborhood $\Omega_1\subset\widetilde\Omega$ of $(q_0,p_0)$ such that 
$$
\forall \delta\in I\,\forall (q,p)\in\Omega: \;p\cdot\nabla_q \widetilde \beta(q,\delta)<0.
$$
The invertibility of $p\cdot\nabla_q V_0(q,\delta)$ in $I\times\Omega_1$ is then implied by
\[
-\det(p\cdot\nabla_q V_0(q,\delta)) = (p\cdot\nabla_q\widetilde\beta(q,\delta))^2 + \delta^2 (p\cdot\nabla_q\widetilde\gamma(q))^2.
\]
\end{proof}

We now work in the set $\Omega_1$ and denote by 
$$
\lambda^\pm(q,\delta) = \alpha(q) \pm\sqrt{\widetilde\beta(q,\delta)^2+\delta^2\widetilde\gamma(q)^2}
$$
the eigenvalues of the matrix $V(q,\delta)$ and still denote by 
$$
\Phi_\pm^t: \R^{2d}\to\R^{2d},\qquad \Phi_\pm^t(q_0,p_0)=\left(q_\delta^\pm(t),p_\delta^\pm(t)\right)
$$ 
the flow associated with the $\delta$-dependent Hamiltonian system 
$$
\dot q_\delta^\pm(t)=p_\delta^\pm(t),\;\; \dot p^\pm_\delta(t)=-\nabla_q \lambda^\pm(q_\delta^\pm(t),\delta),
$$
that becomes singular on the hypersurface $S_0$ if $\delta=0$. However, by the analysis of~\cite[\S3]{FG02} and \cite[\S2]{FG03}, one can pass through the 
singularity in the following sense.  We denote by $H_{\Lambda^\pm}$ the Hamiltonian vector fields of 
$$
\Lambda^\pm(q,p,\delta)= \tfrac{|p|^2}{2} + \lambda^\pm(q,\delta)
$$
and consider $(q,p)\in S_0\times\R^d$ with $p\cdot \nabla V_0(q,0)$ invertible and $\delta=0$. Then, 
\begin{eqnarray}\label{def:H}
\lim_{t},{0^-} H_{\Lambda^+} (\Phi^t_+(q,p))& =&  \lim_{t},{0^+} H_{\Lambda^-} (\Phi^t_-(q,p))\\
\nonumber
&=& p\cdot\nabla_q-\nabla_q \alpha(q)\cdot\nabla_p-\nabla_q\widetilde\beta(q,0)\cdot\nabla_p =:H, \\
\label{def:H'}
\lim_{t},{0^+} H_{\Lambda^+} (\Phi^t_+(q,p)) &=& \lim_{t},{0^-} H_{\Lambda^-} (\Phi^t_-(q,p))\\
\nonumber
&=& p\cdot\nabla_q-\nabla_q \alpha(q)\cdot\nabla_p+\nabla_q\widetilde\beta(q,0)\cdot\nabla_p =:H'.
\end{eqnarray}
Moreover,  the standard symplectic product of $H$ and $H'$ has a sign, since 
$$
\omega(H,H') = 2 p\cdot \nabla_q \widetilde\beta(q,0) <0
$$
according to Lemma~\ref{lem:transverse}. This continuation of the classical trajectories through the crossing at $S_0\times\{\delta=0\}$ will be a crucial element of our analysis. 
  
\begin{remark}\label{rem:beta<0}
For $\delta=0$, on ingoing trajectories, that is, on trajectories entering the conical crossing, we have 
$$
{\d\over \d t} g(q^\pm_{\delta=0}(t),0) = p_{\delta=0}^\pm(t)\cdot \nabla_q g(q_{\delta=0}^\pm(t),0)\leq 0,
$$
which implies that they are included in the set
$$
\left\{\left(p\cdot\nabla_q\widetilde \beta(q,0)\right)\widetilde \beta(q,0)\le 0\right\}\subset\left\{\widetilde\beta(q,0)\ge0\right\}.
$$ 
Similarly,  outgoing trajectories are included in $\left\{\widetilde\beta(q,0)\le 0\right\}$. 
\end{remark}


\subsection{The parametrized Schr\"odinger system}\label{sec:delta}

We analyse the time-dependent Schr\"odinger systems
$$
 \left\{\begin{array}{l}
i\eps\partial_t\psi^\eps_t = -\frac{\eps^2}{2} \Delta_q\psi^\eps_t +V(q,\delta)\psi^\eps_t,\;\;(t,q)\in\R\times\R^d,\\
\psi^\eps_{t=0}=\psi^\eps_0,
\end{array}
\right.
$$
defined by the family of potential matrices $V(q,\delta) = \tfrac12 \tr V(q) + V_0(q,\delta)$ with $\delta\in I$ of Theorem~\ref{prop:parametrization} and equation~\eqref{def:V}.

\medskip
Literally as in Section~\ref{sec:result}, we construct a surface hopping semigroup $({\mathcal L}^t_\eps)_{t\ge0}$ for all $\delta\in I$, and thus obtain an effective dynamical description comprising both the original avoided crossing at $\delta=\delta_0$ and the  conical intersection at $\delta=0$. On the one hand we use classical transport along the flows $\Phi^t_\pm:\R^{2d}\to\R^{2d}$ of
\[
\dot q_\delta^\pm(t) = p_\delta^\pm(t),\qquad \dot p_\delta^\pm(t) = -\nabla\lambda^\pm(q_\delta^\pm(t),\delta). 
\]
On the other hand we monitor the gap function along the classical trajectories and detect local minima by checking whether
$$
\frac{\d}{\d t} g(q^\pm_\delta(t),\delta) = p^\pm_\delta(t) \cdot \nabla_q g(q^\pm_\delta(t),\delta) = 0.
$$
The corresponding jump manifold reads
$$
\Sigma_{\eps} = \left\{(q,p)\in\R^{2d}\mid 
\;g(q,\delta)\le R\sqrt\eps
 ,\;\;p\cdot\nabla_q g(q,\delta)=0\right\}.
$$
The non-adiabatic transition probability for $(q^*,p^*)\in\Sigma_\eps$ is given by 
\begin{equation}\label{Tepsdelta}
T_{\eps}(q^*,p^*,\delta)=\exp\!\left(-\frac{\pi}{4\, \eps} \frac{g(q^*,\delta)^2}{|{\rm det}(p^*\cdot\nabla_qV_0(q^*,\delta))|^{1/2}}\right).
\end{equation}
Finally, the diagonal parts of the Wigner transform $W^\eps(\psi)$ of a wave function $\psi\in L^2(\R^d,\C^2)$ are defined with respect to the eigenprojectors $\Pi^\pm(q,\delta)$, that is, by
$$
w^\eps_\pm(\psi)(q,p,\delta)={\rm tr} \left(\Pi^\pm(q,\delta) W^\eps(\psi)(q,p)\right).
$$ 
The semigroup $(\LL^t_\eps)_{t\ge0}$ then acts on the function $w^\eps(\psi)\in{\mathcal B}$ constructed from the diagonal components 
$w^\eps_\pm(\psi)$ according to relation~\eqref{fpm}.

\medskip
The assumptions $(A0)$, $(A1)$, $(A2)$ of Theorem~\ref{theorem} refer to the potential $V(q)=V(q,\delta_0)$, and we denote by $(A0)_\delta$, $(A1)_\delta$, and $(A2)_\delta$ the corresponding assumptions with respect to $V(q,\delta)$. 
Our aim is to prove the following result:

\begin{theorem}\label{theorem:delta-result}
Let $V$ be a potential matrix with an avoided crossing in the sense of Definition~\ref{def:avoided} and $V(\cdot,\delta)$, $\delta\in I$, the corresponding 
parametrized family of Theorem~\ref{prop:parametrization}. We consider the time-dependent Schr\"odinger equation
$$
i\eps\partial_t\psi^\eps_t = -\tfrac{\eps^2}{2} \Delta_q\psi^\eps_t +V(q,\delta)\psi^\eps_t,\qquad \psi^\eps_{t=0}=\psi^\eps_0,
$$
and assume that $(A0)_\delta$, $(A1)_\delta$ and $(A2)_\delta$ hold for all $\delta\in I$. Then, for all cut-off functions $\chi\in{\mathcal C}^\infty_c([0,T])$, 
there exists a constant $C>0$ such that
\[
\left| \int_0^T \chi(t) \left(w^\eps(\psi^\eps_t)-{\mathcal L^t_\eps w^\eps(\psi^\eps_0)},a\right) \d t \right| \le C\,\eps^{1/32}.
\]
The constant $C$ depends on a finite number of upper bounds of derivatives of the smooth functions $\alpha,\beta,\gamma$ defining the potential $V$ and $a,\chi$ and of lower bounds of the determinant of the matrix $V_1$.
\end{theorem}

The particular choice $\delta=\delta_0$ then implies Theorem~\ref{theorem}.

\begin{remark}\label{rem:delta}
By the construction of Theorem~\ref{prop:parametrization}, the gap function $g(q,\delta)$ has $|\delta|$ as its minimal value. Hence, if  $|\delta|>R\sqrt\eps$, then the jump manifold $\Sigma_\eps$ is empty. In this situation the semigroup $(\mathcal L^t_\eps)_{\eps>0}$ reduces to mere classical transport, that is proven in Appendix~\ref{proof:prop}. 
\end{remark}

\begin{remark}
If the off-diagonal function $\widetilde\gamma$ had derivatives uniformly bounded with respect to the gap parameter $\delta_0$, then Theorem~\ref{theorem:delta-result} would hold with an error of the order $\eps^{1/8}$. We will indicate in Remark~\ref{rem:1/8},~\ref{rem:1/8bis} and~\ref{rem:1/8ter} how the analysis would simplify, if uniform estimates were available.
\end{remark}



\section{Reduction to a Landau--Zener model}\label{sec:reduction}

We now focus on points in the minimal gap hypersurface $S_0$ and construct a  symplectic change of space-time phase space coordinates that allows an elementary microlocal normal form reduction to a Landau--Zener model in \S\ref{sec:normal}. Contrary to the approach in \cite[\S2.8]{CdV1}, the minimal gap size $\delta$ is not treated as another coordinate but as a controlled parameter. 


\subsection{The new symplectic coordinates}\label{sec:coord}
Following the ideas of~\cite[\S6.1]{FG02}, we now construct a symplectic coordinate transformation locally around points in the critical set
\[
S = \left\{(q,t,p,\tau)\in\R^{2d+2}, \; q\in S_0,\; \tau + \tfrac12|p|^2 + \alpha(q) = 0\right\},
\]
that restricts both the energy shells $E^+$ and $E^-$ to the conical crossing situation for $(q,\delta)\in S_0\times\{\delta=0\}$.

\begin{proposition}\label{normalform}
Consider $\rho_0=(q_0,t_0,p_0,\tau_0)\in\R^{2d+2}$ with $q_0\in S_0$ and $p_0\in\R^d$ transverse to $S_0$ at $q_0$. There exists a 
neighborhood $\Omega_2\subset\R^{2d+2}$ of the point $\rho_0$ and for all $\delta\in I$ a positive function $\lambda(\cdot,\delta):\Omega_2\to\R$ such that
$$
\sigma(\rho,\delta)=- \lambda(\rho,\delta)\left(\tau+\tfrac12|p|^2+\alpha(q)\right)\;\;{\rm and} \;\; s(\rho,\delta)=\lambda(\rho,\delta)\widetilde\beta(q,\delta)
$$
satisfy  
$$
\{\sigma(\cdot,\delta),s(\cdot,\delta)\}=1\;\;\text{on}\;\;\Omega_2,
$$
and 
\begin{equation}\label{eq:lambda}
\lambda(\rho,\delta)^2 =-(p\cdot\nabla_q\widetilde\beta(q,\delta))^{-1},\qquad \rho\in S\cap\Omega_2.
\end{equation}
Moreover, all derivatives of $\lambda(\cdot,\delta)$ are uniformly bounded with respect to $\delta\in I$.
\end{proposition}

\begin{proof}
By Lemma~\ref{lem:transverse}, we have for $q\in S_0$ and $p\in\R^d$ transverse to $S_0$ at $q$
$$
\left\{\tau+\tfrac12|p|^2+\alpha(q),\widetilde\beta(q,\delta)\right\}=p\cdot\nabla  \widetilde\beta(q,\delta)<0.
$$
By \cite[Lemma~21.3.4]{Ho}, we can find $\Omega_2$ and a positive function $\lambda(\cdot,\delta)$ such that the functions 
$\sigma(\cdot,\delta)$ and $s(\cdot,\delta)$ satisfy for all $\delta\in I$
$$
\{\sigma(\cdot,\delta),s(\cdot,\delta)\}=1\quad\text{on}\;\;\Omega_2.
$$ 
Indeed, the proof of \cite[Lemma 21.3.4]{Ho} relies on solving differential equations in the variable $(q,p)$, which requires to restrict the set $\Omega_2$. When this is done with coefficients depending smoothly on $\delta$, for $\delta $ in the bounded interval $I$, the restriction can be taken uniformly in $\Omega_2$. Therefore, the set $\Omega_2$ does not depend on~$\delta$. It remains to compute for $\rho\in S\cap\Omega_2$,
$$
1 = \{\sigma(\rho,\delta),s(\rho,\delta)\} = -\lambda(\rho,\delta)^2 \; p\cdot\nabla_q  \widetilde\beta(q,\delta),
$$
and to observe that the derivatives of $\lambda(\cdot,\delta)$ inherit the boundedness of the derivatives of $\widetilde\beta(\cdot,\delta)$, 
see Theorem~\ref{prop:parametrization}.
\end{proof}

We will use this germ of symplectic coordinates and the rotation matrix $R(q)$ introduced in Theorem~\ref{prop:parametrization} to construct a normal form. 
By the Darboux Theorem (see \cite[Theorem 21.1.6]{Ho}) close to  a point $\rho_0=(q_0,t_0,p_0,\tau_0)\in\R^{2d+2}$ with $q_0\in S_0$ and $p_0$ transverse to $S_0$ at $q_0$, there exists a locally defined canonical transform
 $$\kappa_\delta:\;(s,z,\sigma,\zeta)\mapsto (q,t,p,\tau)$$
with $s,\sigma\in\R$ and $(z,\zeta)\in\R^{2d}$, such that 
\begin{equation}\label{def:B}
(RP  R^*)\circ \kappa_\delta = \frac{1}{\lambda\circ\kappa_\delta} \left(-\sigma +
\begin{pmatrix}
s & \delta \check \gamma\\
\delta\check\gamma & -s 
\end{pmatrix}\right),\qquad \check \gamma =( \lambda \widetilde \gamma)\circ \kappa_\delta,
\end{equation}
and the function $\check\gamma$ is nonzero everywhere.

\medskip
This local change of coordinates preserves the symplectic structure of the phase space $\R^{2d+1}_{q,t}\times\R^{2d+1}_{p,\tau}$: The variables $\sigma$ and $\zeta$ are the dual variables of $s$ and~$z$, respectively. Besides, in the new variables $(s,z,\sigma,\zeta)$, 
the geometry of the conical crossing for $\delta=0$ is simple, since  we have 
\[
E^\pm=\{-\sigma\pm \sqrt {s^2+\delta^2\check\gamma^2}=0\},\qquad S=\{ s=0,\;\sigma=0\}.
\]
In particular, by Remark~\ref{rem:beta<0}, the ingoing and outgoing trajectories are included in the sets $\{s\ge0\}$ and $\{s\le0\}$, respectively. The off-diagonal function $\check\gamma$ satisfies additional properties that will be useful later on:

\begin{lemma}\label{lem:gamma}
Let $\delta\in I$ and $\rho\in S\cap\Omega_2$. Then, 
$$
\delta^2 (\check\gamma^2 \circ\kappa^{-1}_\delta)(\rho) = \tfrac14 \,g(q,\delta)^2 \left( |\det p\cdot \nabla_q V_0(q,\delta)|^{-1/2}  + O(\delta^2)\right).
$$
\end{lemma}

\begin{proof}
For $q\in S_0$ we have $\widetilde\beta(q,\delta)=0$ and 
$g(q,\delta)^2= 4 \delta^2 \, \widetilde \gamma(q)^2$. We also observe that
$$
p\cdot\nabla_q\widetilde \gamma(q) =  \{\tfrac{|p|^2}{2},\widetilde\gamma\} = 
-\tfrac{1}{\lambda}\{\sigma,\widetilde\gamma\} - \sigma \{\tfrac{1}{\lambda},\widetilde\gamma\} = 
-\tfrac{1}{\lambda}\partial_s(\widetilde\gamma\circ\kappa_\delta) \circ\kappa_\delta^{-1}
$$
for $\rho=(q,t,p,\tau)\in S$, and that $\partial_s(\widetilde\gamma\circ\kappa_\delta)$ is a derivative normal to $S_0=\{s=0\}$. 
Therefore, by Theorem~\ref{prop:parametrization}, the product $p\cdot\nabla_q\widetilde\gamma(q)$ is bounded. 
Using equation \eqref{eq:lambda}, we then obtain
\begin{align*}
\delta^2 (\check\gamma^2 \circ\kappa^{-1}_\delta)(\rho) &= \delta^2 \lambda(\rho,\delta)^2 \widetilde\gamma(q)^2 = 
-\tfrac14 g(q,\delta)^2 (p\cdot\nabla_q \widetilde\beta(q,\delta))^{-1}\\
&=
\tfrac14 g(q,\delta)^2 \left(\left((p\cdot\nabla_q \widetilde\beta(q,\delta))^2 + \delta^2(p\cdot\nabla_q\widetilde\gamma(q))^2\right)^{-1/2} + O(\delta^2)\right)\\
&=
\tfrac14 \,g(q,\delta)^2 \left( |\det p\cdot \nabla_q V_0(q,\delta)|^{-1/2}  + O(\delta^2)\right).
\end{align*}
\end{proof}


\subsection{The quantization of the normal form}\label{sec:normal}

We now lift the classical normal form~(\ref{def:B}) to the quantum level using  Fourier integral operator theory. 
We define the matrix-valued symbol 
$$
B(q,t,p,\tau,\delta) = \sqrt{\lambda(q,t,p,\tau,\delta)} R(q)
$$ 
for $\delta\in I$ and $(q,t,p,\tau)\in\Omega_2\subset\R^{2d+2}$. The following quantization process retains the minimal gap size $\delta$ as a controlled parameter.

\begin{proposition}\label{prop:quant}
Consider $\rho_0=(q_0,t_0,p_0,\tau_0)\in\R^{2d+2}$ with $q_0\in S_0$ and $p_0\in\R^d$ transverse to $S_0$ at $q_0$. 
Then there exist neighborhoods $\widetilde I\subset I$ of $\delta=0$ and $\Omega_3\subset\Omega_2$ of the point $\rho_0$,  a matrix-valued function $B_\eps=B+\eps B_1$ defined on $\Omega_3$, a canonical transform~$\kappa_\eps$ which is a perturbation of order $\eps$ of the canonical transform $\kappa_{\delta}$, and a unitary operator $K_\eps$ of ${\mathcal L}(L^2(\R^d))$ such that the transformed solution 
$$
v^\eps=K_\eps^* \,\op_\eps(B^*_\eps)^{-1}\psi^\eps_t
$$
satisfies
\begin{equation}\label{eq:systreduit}
\op_\eps(\varphi)\;\op_\eps\!
\begin{pmatrix}
-\sigma + s & \delta \check \gamma \\
\delta\check\gamma& -\sigma -s 
\end{pmatrix}v^\eps  = O(\eps^2),
\end{equation}
for any compactly supported function $\varphi\in{\mathcal C}^\infty_c(\R^{2d+2})$. 
\end{proposition}

Proposition~\ref{prop:quant} allows us to microlocally trade the original Schr\"odinger equation $\op_\eps(P)\psi^\eps_t = 0$ for the reduced system~(\ref{eq:systreduit}).  When saying that the canonical transform $\kappa_\eps$ is  a perturbation of order $\eps$ of $\kappa_\delta$, we mean that $\kappa_\eps$ is defined on a subset $\Omega_3$ of the open set where $\kappa_\delta $ is defined and that for all 
$a\in{\mathcal C}_c^\infty(\kappa_\eps^{-1}(\Omega_3))$, 
$a\circ \kappa_\eps =a\circ\kappa_\delta +O(\eps)$ with respect to the semi-norms of the derivatives of $a$.

\medskip 

We provide the proof of Proposition~\ref{prop:quant} for the sake of completeness. It relies on  the Fourier integral operator construction of~\cite[\S2.2]{FG02}, which is based on Egorov's Theorem, and follows part of the schedule of Colin de Verdi\`ere's normal form approach (see \cite[Theorem~3]{CdV1} and also~\cite[Theorem~1]{Fe06}).

\begin{proof}
The first step uses~\cite[\S2.2]{FG02}: There exists a Fourier integral operator $K_0$ associated with $\kappa_\delta$ such that 
for any $a\in{\mathcal C}_c^\infty(\Omega_2, \C^{2\times 2})$,
$$
K_0^* \op_\eps(a) K_0- \op_\eps(a\circ\kappa_\delta)  = O(\eps^2)
$$
in ${\mathcal L}(L^2(\R^{d+1}))$, 
where the $O(\eps^2)$ contains semi-norms of derivatives of $a$.  We note that the the Fourier integral operator is a diagonal operator with the same scalar operator on each position of the diagonal. 

\medskip
The second steps turns to the classical normal form \eqref{def:B} to obtain 
\[
K_0^*\op_\eps(B P B^*) K_0=\op_\eps\!\begin{pmatrix}-\sigma+s & \delta\check\gamma\\ \delta \check\gamma & -\sigma-s\end{pmatrix}+ O(\eps^2).
\]
This suggests the change of unknown $\psi^\eps \mapsto K_0^*\op_\eps(B^*)^{-1}\psi^\eps$. However, this change of unknown generates unsatisfactory terms of order $\eps$, since we have 
by symbolic calculus
$$\op_\eps(BPB^*)=\op_\eps(B)\op_\eps(P)\op_\eps(B^*) + \eps\, \op_\eps(R) +O(\eps ^2),$$
where $R$ is the self-adjoint matrix defined by
$$R= {1\over 2i}( B\{P,B^*\} + \{B,P\}B^*).$$
More precisely, we have 
\[
K_0^*\op_\eps(B)\op_\eps( P)\op_\eps( B^*) K_0=\op_\eps\!\begin{pmatrix}-\sigma+s & \delta\check\gamma\\ \delta \check\gamma & -\sigma-s\end{pmatrix} +\eps\,\op_\eps(R\circ\kappa_\delta)+ O(\eps^2).
\]
We note that here the function $\delta\check\gamma$ is treated as a whole. Since it has bounded derivatives, the above remainder is uniform with respect to $\delta$.
Next we will remove the term of order $\eps$ by modifying $B$ and $\kappa_\delta$ at order~$\eps$, which will give the stated change of unknown $v^\eps=K_\eps^*\op_\eps(B^*_\eps)^{-1} \psi^\eps$. 

\medskip
In the rest of this proof, and only here, $\tau$ will be a parameter belonging to $[0,1]$. We define the canonical transform $\kappa_1(\tau)$ which is a perturbation of identity and solves the Hamiltonian equation 
$$\frac{\d}{\d\tau} \kappa_1(\tau ) =H_{1+\eps\varphi} \kappa_1(\tau),\;\; \kappa_1(0)={\rm Id},$$
where $\varphi\in{\mathcal C}_c^\infty(\R^{2d+2})$ is a smooth function that we shall fix later on.  The canonical transform $\kappa_1$ is a perturbation of order $\eps$ of the identity, so that 
$$\kappa_\eps :=\kappa_\delta\circ\kappa_1(1)$$
 is the sought-after perturbation of order $\eps$ of $\kappa_\delta$.
 We associate with $\kappa_\eps(\tau)$ a Fourier integral operator $K_\eps(\tau)$ by setting 
 $$i\eps \frac{\d}{\d\tau} K_\eps(\tau) = \op_\eps (1+\eps \varphi) K_\eps (\tau),\;\; K_\eps (0)={\rm Id}.$$
 The solution of our problem will be 
 $$K_\eps:=K_0\circ K_\eps(1).$$
 Note that this construction method \cite[\S2.2]{FG02} has been used for $K_0$: Given a canonical transform, one links it to the identity in a differentiable way,
thereby defining a function $\varphi$ 
and the operator $K_\eps(\tau)$ as a solution of a differential system. 

\medskip
 We define the matrix $B_\eps(\tau)=B+\eps \tau B_1$ where $B_1$ will be fixed later, such
 $$B_\eps:=B_\eps(1),$$
 will be the solution of the Proposition.
 Let us now investigate how $B_1$ and $\varphi$ have to be chosen, which might require to restrict to smaller neighbourhood of $\delta=0$ and $\rho_0$. 
 For $\tau\in (0,1)$, we set 
 $$L_\eps(\tau):= K_\eps(\tau)^* 
 K_0^* \left[\op_\eps(B_\eps(\tau))\op_\eps (P) \op_\eps (B_\eps(\tau)^*)-\eps\,(1-\tau) \op_\eps(R)
 \right]K_0 K_\eps(\tau).$$
 We have 
 $$
 L_\eps(0)=\op_\eps\!\begin{pmatrix}-\sigma+s & \delta\check\gamma\\ \delta \check\gamma & -\sigma-s\end{pmatrix} + O(\eps^2)
 $$
  and we are going to prove that we can find $B_1$ and $\varphi$ such that $\frac{\d}{\d\tau} L_\eps(\tau)=O(\eps^2)$, so that we shall get $L_\eps(1)=L_\eps(0)+ O(\eps^2)$ which will conclude our proof. 

\medskip
A simple computation shows that 
 \begin{eqnarray*}
\frac{\d}{\d\tau} L_\eps(\tau) & = & \eps\,  K_\eps(\tau)^*  K_0^* \Bigl[
 \op_\eps(B_1P B_0^*+B_0 P B_1^*)+ \op_\eps(R) \\ 
 &&+ {1\over i} \left[\op_\eps (\varphi), \op_\eps(B_0 PB_0^*)\right]
\Bigr] K_0 K_\eps(\tau)  +O(\eps^2)\\
 & = &  \eps \, K_\eps(\tau)^*  K_0^* \op_\eps\left[ B_1 PB_0^* + B_0 P B_1^* + R + \{\varphi, B_0 P B_0^*\}
\right] K_0 K_\eps(\tau)\\ 
&&  +\,O(\eps^2) .
 \end{eqnarray*}
The choice of $B_1$ and $\varphi$ such that $B_1 PB_0^* + B_0 P B_1^* + R + \{\varphi, B_0 P B_0^*\}=0$ is possible by~\cite[Lemma~5]{CdV1}. 
\end{proof}

\subsection{The off-diagonal components}
Our final step towards the Landau--Zener model is to remove the dependence of  the off-diagonal function $\check\gamma(s,z,\sigma,\zeta,\delta)$ on the coordinates $s$ and $\sigma$, following the method proposed in~\cite[Lemma~5 and 6]{FG03}, see also~\cite[Proposition~8]{FG02}.
From now on, we restrict ourselves to 
\[
0<\delta\le R\sqrt\eps,\qquad R=R(\eps) = \eps^{-1/8},
\] 
see also Remark~\ref{rem:delta}. Moreover, 
since the scattering result for the Landau--Zener system, that we use in Section~\ref{sec:LZ_form}, has an error estimate of the order $R^2\sqrt\eps/s$, 
we also start focusing on regions, where 
\[
\tfrac12 r R^2\sqrt\eps \le |s| \le r R^2\sqrt\eps 
\] 
for some $1\ll r \le R$, that will be chosen as $r = r(\eps) = \eps^{-1/32}$ later on. These choices of $r$ and $R$ imply that 
$\delta\le |s|$ as soon as $\eps^{5/32}\le 1/2$.

\begin{lemma}\label{rid_of_s_and_sigma}
On $\kappa_\delta(\Omega_3)\times\widetilde I$, there
exist matrix-valued functions $M^\eps(s,z,\sigma,\zeta,\delta)$ and $\widetilde M^\eps(s,z,\sigma,\zeta,\delta)$, such that 
$$
M^\eps = M_0^\eps + \delta M_1^\eps,\qquad \widetilde M^\eps = \widetilde M_0^\eps + \delta \widetilde M_1^\eps, 
$$
with $M_0^\eps$ and $\widetilde M_0^\eps$ unitary matrices, and for all $\varphi\in\mathcal C^\infty_c(\R^{2d+2})$ supported in a set 
with $s=O(r R^2\sqrt\eps)$, with $1\ll r\leq R$,  one has
\begin{align*}
& \op_\eps(\varphi)\,
\op_\eps(\widetilde M^\eps)\,  \op_\eps\!\begin{pmatrix}-\sigma+ s & \delta \check \gamma \\\delta\check\gamma& -\sigma-s 
\end{pmatrix} =\\
&\qquad \op_\eps(\varphi)\, \op_\eps\!\begin{pmatrix}
-\sigma + s & \delta \check \gamma_0 \\
\delta\check\gamma_0& -\sigma -s 
\end{pmatrix} \op_\eps(M^\eps) + O(r^2\eps^{7/8})
\end{align*}
in ${\mathcal L}(L^2(\R^{d+1}))$, where 
$$
\check\gamma_0(z,\zeta,\delta) = \check\gamma(0,z,0,\zeta,\delta).
$$ 
Moreover, the families $(\op_\eps(\varphi)\, \op_\eps(M^\eps))_{\eps,\delta>0}$ and $(\op_\eps(\varphi)\, \op_\eps(\widetilde M^\eps))_{\eps,\delta>0}$ are uniformly bounded in ${\mathcal L}(L^2(\R^{d+1}))$.
\end{lemma}

\begin{remark}\label{rem:1/8}
The proof below shows that if $\check\gamma(\cdot,\delta)$ had bounded derivatives uniformly with respect to $\delta$, then Lemma~\ref{rid_of_s_and_sigma} would hold with a remainder estimate of the order $\eps\delta$.
\end{remark}

\begin{remark}\label{rem:CVeps}
To estimate the norm of operators such as $\op_\eps(\check\gamma)$, 
 we shall use  the scaling operator~$T_\eps$ defined by 
$$\forall f\in L^2(\R^{d+1}),\;\; T_\eps f(u)=\eps^{d+1\over 4} f(\sqrt\eps u).$$
This unitary operator is such that 
$$\forall a\in{\mathcal C}_0^\infty(\R^{2d+2}),\;\; T_\eps\op_\eps(a) T_\eps^* =\op_1(a(\sqrt\eps\cdot,\sqrt\eps\cdot)).$$
The Calder\'on--Vaillancourt theorem yields the existence of $N\in\N$ and $C_N>0$ such that
$$\forall a\in{\mathcal C}_c^\infty(\R^{2d+2}),\;\; \| \op_\eps(a)\| _{{\mathcal L}(L^2(\R^{d+1}))}\leq C_N \, \sup_{\beta\in\N^{2d+2},\;\;|\beta|\leq N}\,\sup_{\rho\in\R^{2d+2}} \left(\eps^{|\beta|\over 2} \left| \partial_\rho ^\beta a(\rho) \right| \right)$$
holds. As a consequence,  Theorem~\ref{prop:parametrization} implies for all $\varphi\in\mathcal C^\infty_c(\R^{2d+2})$
$$\|\op_\eps (\varphi) \op_\eps( \check\gamma)\|_{{\mathcal L}(L^2(\R^{d+1}))}\leq C (1 +\sqrt\eps\,  \delta^{-1}\sup_{s\in{\rm supp}(\varphi)} |s|).$$
\end{remark}

\begin{proof}
We consider the three matrices 
$$J=\begin{pmatrix} 1 & 0 \\ 0 & -1\end{pmatrix},\;\;
K=\begin{pmatrix} 0 & 1 \\ -1 & 0\end{pmatrix},\;\;
L=\begin{pmatrix} 0 & 1 \\ 1 & 0\end{pmatrix},$$
that satisfy
$$
JL = K = -LJ,\qquad JK = L = -KJ,\qquad J^2 = \Id.
$$
We write
$$
\begin{pmatrix}
-\sigma + s & \delta\check\gamma\\ \delta\check\gamma & -\sigma -s 
\end{pmatrix} = -\sigma \Id + s J + \delta\check\gamma L
$$
and proceed in two steps.

\medskip 

We first remove the $s$-dependence of $\check\gamma$ and construct matrix-valued functions $D_j^\eps=D_j^\eps(s,z,\sigma,\zeta,\delta)$, $j=0,1$, with
\begin{align*}
& \op_\eps(D_0^\eps+\delta D_1^\eps)\,\op_\eps(-\sigma\Id + s J + \delta \check\gamma L)\\
&\qquad = \op_\eps(-\sigma\Id + s J + \delta \check\gamma_* L) \,\op_\eps(D_0^\eps + \delta D_1^\eps) + o(\eps),
\end{align*}
where 
$$
\check\gamma_*(z,\sigma,z,\delta) = \check\gamma(0,z,\sigma,\zeta,\delta).
$$
Symbolic calculus provides that the above equation is equivalent to
\begin{align*}
& \op_\eps\!\left( (D_0^\eps+\delta D_1^\eps)(-\sigma\Id + s J + \delta \check\gamma L)\right)
+ \frac{\eps}{2i}\op_\eps\!\left( \{ D_0^\eps+\delta D_1^\eps,-\sigma\Id + s J + \delta \check\gamma L\} \right)\\
&= \op_\eps\!\left( (-\sigma\Id + s J + \delta \check\gamma_* L) (D_0^\eps + \delta D_1^\eps)\right)
+ \frac{\eps}{2i} \op_\eps\!\left( \{-\sigma\Id + s J + \delta \check\gamma_* L,D_0^\eps + \delta D_1^\eps\} \right)\\
&\qquad + \op_\eps(\rho^\eps).
\end{align*}
where the remainder symbol $\rho^\eps$ consists of second order derivatives terms times a factor $\eps^2$. 
Since any derivatives of $\delta\check\gamma$ and $\delta\check\gamma_*$ are of the order $|s|$, 
Remark~\ref{rem:CVeps} yields that
\[
\op_\eps(\varphi) \op_\eps(\rho^\eps) = O(|s|\cdot |s|^2 \delta^2) = O(r^3 \,\eps^{3/2}),
\]
if the a priori estimate 
\begin{equation}\label{apriori}
\partial^{\alpha}(D_0^\eps + \delta D_1^\eps) = O((|s|\delta/\eps)^{|\alpha|}),\qquad \alpha\in\N^{2d+2},
\end{equation}
holds, that we will justify later during the proof. We neglect the term
\begin{align*}
&\frac{\eps}{2i} \{D^\eps_0,\delta\check\gamma L\} + \frac{\eps}{2i} \{\delta D^\eps_1,-\sigma \Id + sJ + \delta \check\gamma L\} \\
&-\frac{\eps}{2i} \{\delta\check\gamma_* L,D^\eps_0\} - \frac{\eps}{2i} \{-\sigma \Id + sJ + \delta \check\gamma_* L,\delta D^\eps_1\}, 
\end{align*}
which by the same argument produces an error of the order
\[
O(|s|\cdot|s|\delta) = O(r^2 \eps^{7/8})
\]
in the region of observation. Then, we obtain the three relations
\begin{align*}
& D_0^\eps(-\sigma\Id + sJ) = (-\sigma\Id + sJ) D_0^\eps,\\
& D_1^\eps(-\sigma\Id + sJ) +\check\gamma D_0^\eps L = (-\sigma\Id + sJ) D_1^\eps + \check\gamma_* L D_0^\eps,\\
& \check\gamma D_1^\eps L + \frac{\eps}{2i\delta^2}\{D_0^\eps,-\sigma \Id + sJ \} = 
\check\gamma_* L D_1^\eps + \frac{\eps}{2i\delta^2}\{-\sigma \Id + sJ, D_0^\eps\}.
\end{align*}
We make the ansatz
$$
D_0^\eps = \begin{pmatrix}\widetilde a_0^\eps & 0\\ 0& \widetilde d_0^\eps\end{pmatrix},\qquad
D_1^\eps = \begin{pmatrix}0 & \widetilde b_1^\eps\\ \widetilde c_1^\eps & 0\end{pmatrix},
$$
and rewrite the second of the three relations as
\begin{align*}
s\begin{pmatrix}0 & -2\widetilde b_1^\eps\\ 2\widetilde c_1^\eps & 0\end{pmatrix} = 
\begin{pmatrix}0 & -\check\gamma \widetilde a_0^\eps+ \check\gamma_*\widetilde d_0^\eps\\ 
-\check\gamma \widetilde d_0^\eps+ \check\gamma_* \widetilde a_0^\eps & 0\end{pmatrix}.
\end{align*}
This requires 
$$
\widetilde a_0^\eps(0,z,\zeta,\delta) = \widetilde d_0^\eps(0,z,\zeta,\delta)
$$ 
and 
$$
\widetilde b_1^\eps = \frac{1}{2s}(\check\gamma\widetilde a_0^\eps - \check\gamma_*\widetilde d_0^\eps),\qquad
\widetilde c_1^\eps = \frac{1}{2s}(\check\gamma_*\widetilde a_0^\eps - \check\gamma\widetilde d_0^\eps).
$$
The third relation can be rewritten as
$$
\begin{pmatrix}\check\gamma \widetilde b_1^\eps -\check\gamma_* \widetilde c_1^\eps & 0 \\ 0 & \check\gamma \widetilde c_1^\eps -  \check\gamma_* \widetilde b_1^\eps\end{pmatrix} = \frac{\eps}{i\delta^2}\begin{pmatrix}  -\partial_s\widetilde a_0^\eps -\partial_\sigma \widetilde a_0^\eps & 0 \\ 0 & 
-\partial_s\widetilde d_0^\eps + \partial_\sigma \widetilde d_0^\eps\end{pmatrix}.
$$
We define 
\begin{align*}
\widetilde\vartheta^\eps(s,z,\sigma,\zeta,\delta) &= \frac{i\delta^2}{2\eps s}
\left(\check\gamma_*^2(z,\sigma,\zeta,\delta) - \check\gamma^2(s,z,\sigma,\zeta,\delta)\right)\\
& =- \frac{i\delta^2}{2\eps }\int_0^1 \partial_s (\check \gamma^2) (sr,z,\sigma,\zeta,\delta) dr
\end{align*}
and observe that all $s$-derivatives of $\widetilde\vartheta^\eps$ are of the order $\delta^2/\eps$, while any derivative with respect to $(z,\sigma,\zeta)$ is of the order $\delta/\eps$ in view of Theorem~\ref{prop:parametrization}. We
obtain the equations
$$
(\partial_s +\partial_\sigma) \widetilde a_0^\eps = \widetilde\vartheta^\eps \, \widetilde a_0^\eps,\qquad
(\partial_s-\partial_\sigma) \widetilde d_0^\eps = -\widetilde\vartheta^\eps \, \widetilde d_0^\eps,
$$
that can be solved by
\begin{align*}
\widetilde a_0^\eps(s,z,\sigma,\zeta,\delta) &= \exp\!\left(\int_0^s \widetilde\vartheta^\eps(\tau,z,\sigma-s+\tau,\zeta,\delta) \d\tau\right),\\
\widetilde d_0^\eps(s,z,\sigma,\zeta,\delta) &= \exp\!\left(-\int_0^s \widetilde\vartheta^\eps(\tau,z,\sigma+s-\tau,\zeta,\delta) \d\tau\right)
\end{align*}
such that $\widetilde a_0^\eps(0,z,\sigma,\zeta,\delta) = \widetilde d_0^\eps(0,z,\sigma,\zeta,\delta) = 1$.
We observe that 
\[
\partial^\alpha_{z,\sigma,\zeta} \,\widetilde a^\eps_0,\, \partial^\alpha_{z,\sigma,\zeta}\, \widetilde d^\eps_0 = O( (|s|\delta/\eps)^{|\alpha|})
\]
for any $\alpha\in\N^{2d+1}$, while the $s$-derivatives satisfy
\[
\partial_s^k \,\widetilde a^\eps_0, \, \partial^k_s \,\widetilde d^\eps_0 = O(\delta^2/\eps) + O(|s|\delta/\eps)
\]
for any $k\ge1$. We now write
\[
\delta \widetilde b^\eps_1 = \frac{\delta}{2} \,\widetilde a^\eps_0 \int_0^1 \partial_s (\check \gamma) (sr,z,\sigma,\zeta,\delta) dr  
+ \frac{\delta}{2s}\gamma_*(\widetilde a^\eps_0 - d^\eps_0)
\]
and derive a similar expression for $\delta\widetilde c^\eps_1$, such that 
\[
\partial^\beta( \delta D^\eps_1) = O((|s|\delta/\eps)^{|\beta|}) + O((\delta^2/\eps)^{|\beta|})
\]
for all $\beta\in\N^{2d+2}$.This implies the claimed a priori estimate \eqref{apriori}.

\medskip
We now remove the $\sigma$-dependence of the scalar function $\check\gamma_*$, taking advantage of the boundedness of any derivatives of the function $\check\gamma_* = \check\gamma \mid_{\{s=0\}}$, see Theorem~\ref{prop:parametrization}. We look for two matrix-valued functions 
$C_j ^\eps= C_j^\eps(\sigma,z,\zeta,\delta)$, $j=0,1$, with the following properties. First, they satisfy the intertwining relation
\begin{align*}
&\op_\eps(J(C_0^\eps + \delta C_1^\eps) J)\;\op_\eps(-\sigma\Id + s J + \delta \check\gamma L)\\ 
&\qquad = \op_\eps(-\sigma\Id + sJ + \delta \check\gamma_* L) \;\op_\eps(C_0^\eps + \delta C_1^\eps) + o(\eps)
\end{align*}
in $\mathcal{L}(L^2(\R^{d+1}))$, 
where 
$$
\check\gamma_0(s,z,\zeta,\delta) = \check\gamma_*(z,0,\zeta,\delta) = \check\gamma(0,z,0,\zeta,\delta).
$$
This relation is equivalent to 
\begin{align*}
&\op_\eps(C_0^\eps + \delta C_1^\eps)\;\op_\eps(-\sigma J + s \Id + \delta \check\gamma_* K)\\
&\qquad = \op_\eps(-\sigma J + s\Id + \delta \check\gamma_0 K) \;\op_\eps(C_0^\eps + \delta C_1^\eps) + O(\eps^2\delta),
\end{align*}
using the growth bound
$$
\forall\alpha\in\N^{2d+2}\,  \exists c_\alpha>0\, \forall \eps,\delta>0:  \|\partial^\alpha C_j^\eps(\cdot,\delta)\|_\infty < c_\alpha\, (\delta^2/\eps)^{|\alpha|}.
$$
Symbolic calculus yields
\begin{align*}
&\op_\eps\!\left( (C_0^\eps + \delta C_1^\eps)(-\sigma J + s \Id + \delta \check\gamma_* K)\right)
+ \frac{\eps}{2i}\,\op_\eps\!\left(\{ C_0^\eps + \delta C_1^\eps, -\sigma J + s \Id + \delta \check\gamma_* K\}  \right)\\
& = \op_\eps\!\left((-\sigma J + s\Id + \delta \check\gamma_0 K)(C_0^\eps + \delta C_1^\eps)\right)
+ \frac{\eps}{2i} \,\op_\eps\!\left(\{-\sigma J + s\Id + \delta \check\gamma_0 K,C_0^\eps + \delta C_1^\eps\} \right)\\
& \qquad + O(\eps^2\delta),
\end{align*}
where the neglected terms of the form $\eps^2$ times derivatives of the order $\ge 2$ define the $O(\eps^2\delta)$ remainder. 
We now sort in powers of $\delta$ and obtain the following three relations, 
\begin{align}\label{eq:delta0}
& C_0^\eps(-\sigma J + s\Id) = (-\sigma J+s\Id) C_0^\eps,\\\label{eq:delta1}
& C_1^\eps(-\sigma J + s\Id) +\check\gamma_* C_0^\eps K = (-\sigma J + s\Id) C_1^\eps + \check\gamma_0KC_0^\eps,\\\label{eq:delta2}
& \check\gamma_* C_1^\eps K + \frac{\eps}{2i \delta^2}  \{C_0^\eps,-\sigma J + s\Id\} = 
\check\gamma_0 KC_1^\eps + \frac{\eps}{2i\delta^2} \{-\sigma J + s\Id, C_0^\eps\}.
\end{align}
We denote the components of the two matrices  by
$$
C_0^\eps = \begin{pmatrix}a_0^\eps & b_0^\eps\\ c_0^\eps & d_0^\eps\end{pmatrix},\qquad
C_1^\eps = \begin{pmatrix}a_1^\eps & b_1^\eps\\ c_1^\eps & d_1^\eps\end{pmatrix}.
$$
The first relation \eqref{eq:delta0} is equivalent to
$[J,C_0^\eps] = 0$, that is, 
$$
b_0^\eps = c_0^\eps = 0.
$$
The second relation \eqref{eq:delta1} is equivalent to $-\sigma [C_1^\eps,J]  = -\check\gamma_* C_0^\eps K + \check\gamma_0 K C_0^\eps$, that is, 
$$
\sigma \begin{pmatrix} 0 & 2b_1^\eps\\ -2c_1^\eps & 0\end{pmatrix} = 
\begin{pmatrix} 0 & -\check\gamma_* a_0^\eps +\check\gamma_0 d_0^\eps\\ \check\gamma_* d_0^\eps -\check\gamma_0 a_0^\eps& 0\end{pmatrix}.
$$
This requires
$$
a_0^\eps(z,0,\zeta,\delta) = d_0^\eps(z,0,\zeta,\delta)
$$
and
$$
b_1^\eps = \frac{1}{2\sigma}\left(-\check\gamma_* a_0^\eps +\check\gamma_0 d_0^\eps\right),\qquad
c_1^\eps = \frac{1}{2\sigma}\left(- \check\gamma_* d_0^\eps + \check\gamma_0 a_0^\eps\right).
$$
The third relation \eqref{eq:delta2} is equivalent to
$$
\check\gamma_* C_1^\eps K - \check\gamma_0KC_1^\eps = 
\frac{\eps}{i\delta^2}\left(-\partial_\sigma C_0^\eps - \tfrac12 (J\partial_s C_0^\eps + \partial_s C_0^\eps J)\right), 
$$
that is,
$$
\check\gamma_* \begin{pmatrix} -b_1^\eps & a_1^\eps\\ -d_1^\eps & c_1^\eps\end{pmatrix} 
- \check\gamma_0 \begin{pmatrix} c_1^\eps & d_1^\eps\\ -a_1^\eps & -b_1^\eps\end{pmatrix}
= \frac{\eps}{i\delta^2}\begin{pmatrix}-\partial_\sigma a_0^\eps  & 0\\ 0 & -\partial_\sigma d_0^\eps\end{pmatrix}.
$$
This can be satisfied by 
$$
a_1^\eps=d_1^\eps=0,
$$ 
and requires  
\begin{align*}
\frac{a_0^\eps}{2\sigma} \left(\check\gamma_*^2-\check\gamma_0^2\right) &= -\frac{\eps}{i\delta^2} \partial_\sigma a_0^\eps ,\\
\frac{d_0^\eps}{2\sigma} \left(\check\gamma_0^2-\check\gamma_*^2\right) &= -\frac{\eps}{i\delta^2} \partial_\sigma d_0^\eps.
\end{align*}
We set
$$
\vartheta^\eps(z,\sigma,\zeta,\delta) = \frac{i\delta^2}{2\eps \sigma} \left(\check\gamma_0^2(z,\zeta,\delta)-\check\gamma_*^2(z,\sigma,\zeta,\delta)\right)
$$
and rewrite the above equations as
$$
\vartheta^\eps a_0^\eps = \partial_\sigma a_0^\eps,\qquad \vartheta^\eps d_0^\eps = -\partial_\sigma d_0^\eps.
$$
The functions
\begin{align*}
a_0^\eps(z,\sigma,\zeta,\delta) &= \exp\!\left(\int_0^{\sigma} \vartheta^\eps(z,\tau,\zeta,\delta)\, \d\tau\right),\\
d_0^\eps(z,\sigma,\zeta,\delta) &= \exp\!\left(\int_0^{\sigma} \vartheta^\eps(z,\tau,\zeta,\delta)\, \d\tau\right)
\end{align*}
solve these equations and satisfy 
$$
a_0^\eps(z,0,\zeta,\delta) = d_0^\eps(z,0,\zeta,\delta) = 1.
$$ 
We conclude that the constructed matrices $C_0^\eps$ and $C_1^\eps$ have the desired properties. 
\end{proof}


\subsection{Arriving at the Landau--Zener model}\label{sec:arrival}
We now use Proposition~\ref{prop:quant} and Lemma~\ref{rid_of_s_and_sigma} to introduce
$$
\widetilde v^\eps = \op_\eps(M^\eps) v^\eps\quad\text{with}\quad M^\eps = M_0^\eps + \delta M_1^\eps.
$$
Since
\begin{align*}
\op_\eps\!\begin{pmatrix}-\sigma+s & \delta\check\gamma_0\\ \delta\check\gamma_0 & -\sigma -s\end{pmatrix}\widetilde v^\eps = 
\op_\eps(\widetilde M^\eps) \!\begin{pmatrix}-\sigma+s & \delta\check\gamma\\ \delta\check\gamma & -\sigma -s\end{pmatrix}v^\eps + O(r^2\eps ^{7/8}),
\end{align*}
we obtain for all $\varphi\in{\mathcal C}^\infty_c(\R^{2d+2})$ the doubly reduced system
\begin{equation}\label{eq:L2sz}
\op_\eps(\varphi) \;\op_\eps\!\begin{pmatrix}-\sigma+s & \delta\check\gamma_0\\ \delta\check\gamma_0 & -\sigma -s\end{pmatrix}\widetilde v^\eps =
O(r^2\eps ^{7/8}) \;\;\text{in}\;\;L^2(\R^{d+1}).
\end{equation}
The estimate of \cite[Proposition~7]{FG02} also implies, that $(\op_\eps(\varphi)\widetilde v^\eps)_{\eps>0}$ is a bounded sequence in 
$L^\infty(\R_s,L^2(\R^d_z))$, and we compare the new function $\widetilde v^\eps$ with the solution~$\check v^\eps$ of the Landau--Zener type system 
\begin{equation}\label{eq:systreduitbis}
{\eps\over i} \partial_s \check v^\eps = \op_\eps
\begin{pmatrix}
s & \delta \check \gamma_0\\
\delta\check\gamma_0 & -s 
\end{pmatrix}\check v^\eps,\qquad\check v^\eps|_{s=0}\,=\,\widetilde v^\eps|_{s=0}.
\end{equation}
The order $r^2\eps ^{7/8}$ right hand side of the doubly reduced system \eqref{eq:L2sz} is small enough to be treated as a perturbation 
and we obtain a positive constant $C>0$ such that for all $s\in\R$ and $\eps>0$,
\begin{equation}\label{eq:comparaison}
\left\|\widetilde v^\eps(s) - \check v^\eps(s)\right\|_{L^2(\R^d_z)} \le C |s| r^2\eps ^{-1/8}.
\end{equation}
This implies that for all $\varphi\in \mathcal C_c^\infty(\R^{2d+2})$ that are supported in a region, where $s\sim r R^2\sqrt\eps = r\eps^{1/4}$, we obtain
\begin{equation}\label{eq:error}
\op_\eps(\varphi)\left(\widetilde v^\eps - \check v^\eps\right) = O(r^3\eps ^{1/8})\;\;\text{in}\;\;L^2(\R^{d+1}_{s,z}),
\end{equation}
and we have established the microlocal link of the original Schr\"odinger equation
$$\
\op_\eps(P)\psi^\eps_t = 0
$$
to the Landau--Zener system~\eqref{eq:systreduitbis},  provided we choose $r \le \eps^{ -\kappa}$ with $0<\kappa<1/24$.

\begin{remark}\label{rem:1/8bis}
If $\check\gamma(\cdot,\delta)$ had uniformly bounded derivatives, then,  in view of Remark~\ref{rem:1/8}, the term $ |s| r^2\eps ^{-1/8}$ in~(\ref{eq:comparaison}) would be replaced by $\delta|s|$.  For bounded times, this remainder would be of order $\eps^{3/8}$.
\end{remark}


\section{The proof of the main result}\label{sec:mainproof}

For proving our main result Theorem~\ref{theorem}, we now analyse the dynamics of the block diagonal components of the Wigner transform $w^\eps_\pm(\psi^\eps_t)$, that is, 
\[
\int_{\R^{2d+1}} \chi(t) a(q,p) \,w^\eps_\pm(\psi^\eps_t)(q,p,\delta)\, \d(q,p,t)
= \left\langle\op_\eps(\chi a\Pi^\pm)\psi^\eps_t,\psi^\eps_t\right\rangle
\]
for scalar observables $a\in{\mathcal C}_c^\infty(\R^{2d})$ with $\supp(a)\subset\Omega_3$ and $\chi\in{\mathcal C}^\infty_c([0,T])$. By the assumption (A1)$_\delta$,  the observables have support away from the set of small eigenvalue gap, providing the derivative bounds
\begin{equation}\label{eq:Pi}
\forall \alpha,\beta\in\N^{d}\,\exists c_{\alpha,\beta}>0\,\forall \eps,\delta>0: 
\left\| \partial^\beta_q\partial^\alpha_p(a\Pi^\pm(\cdot,\delta)) \right\|_\infty < c_{\alpha,\beta} (R^3\sqrt\eps)^{-|\beta|}.
\end{equation}
In the following, the scale $R\sqrt\eps$ will play an important role. We note that, of course, $(R^3\sqrt\eps)^{-|\beta|}\le (R\sqrt\eps)^{-|\beta|}$ for $\beta\in\N^d$.
Our first step in the proof is now the replacement of the eigenprojectors $\Pi^\pm(q,\delta)$ in $\op_\eps(\chi a \Pi^\pm)\psi^\eps_t$ by the localisation on the corresponding energy shell
$$
E^\pm=\left\{(q,t,p,\tau)\in\R^{2d+2},\;\;\tau  + \Lambda^\pm (q,p,\delta)=0\right\},
$$
that is a subset of the space-time phase space.


\subsection{Localization in energy}\label{sec:energy}
The localization is implemented by a smooth cut-off function $\theta\in{\mathcal C}_c^\infty(\R)$ satisfying
\begin{align*}
0\leq \theta(u)\leq 1,\quad 
\theta(u)=0\;\text{for}\;|u|>1,\quad
\theta(u)=1\;\text{for}\;|u|<1/2.
\end{align*}
We combine it with the energy function and set
$$
\theta^\pm_{\eps,R} (q,p,\tau,\delta)= \theta\!\left({\tau + \Lambda^\pm(q,p,\delta )\over
R\sqrt\eps}\right),\qquad (q,p,\tau)\in\R^{2d+1}.
$$

\begin{lemma}\label{lem:loc}
For all symbols $\chi\in{\mathcal C}_c^\infty(\R)$ and $a\in{\mathcal C}^\infty_c(\R^{2d})$ satisfying the derivative bounds \eqref{eq:Pi}, we have 
\[
\op_\eps\!\left(\chi a \Pi^\pm\right)\psi^\eps_t = \op_\eps\!\left(\chi a \theta^\pm_{\eps,R}\right)\psi^\eps_t + O(R^{-2}) + O(R^{-1}\sqrt\eps) 
\]
in $L^2(\R^{d+1}_{t,q})$.
\end{lemma}

\begin{proof} Following the lines of the proof of \cite[Lemma~5.1]{FL08}, we observe that, since $1-\theta$ vanishes identically close to $0$, we can write 
$$1-\theta^\pm_{\eps,R}(q,p,\tau,\delta)= \frac{\tau +\Lambda^\pm (q,p,\delta)}{R\sqrt\eps} G\!\left({\tau +
\Lambda^\pm (q,p,\delta)\over
R\sqrt\eps}\right)$$
for some smooth function $G$. Since 
$$
( \tau + \Lambda^\pm(q,p,\delta)) \Pi^\pm(q,\delta)=\Pi^\pm(q,\delta )P(q,p,\tau,\delta),
$$
we have
\[
(1-\theta^\pm_{\eps,R})\Pi^\pm
 =  \frac{1}{R\sqrt\eps} \,G\!\left(\frac{\tau+\Lambda^\pm}{R\sqrt\eps}\right)\Pi^\pm P.
\]
We can use symbolic calculus to bring into play that $\psi^\eps_t$ solves the Schr\"odinger equation 
$\op_\eps(P)\psi^\eps_t = 0$.  
The derivative bounds \eqref{eq:Pi} imply that
\[
\op_\eps\!\left(\chi a(1-\theta^\pm_{\eps,R})\Pi^\pm\right) \psi^\eps_t = O(R^{-2}) + O(R^{-1}\sqrt\eps)\quad\text{in}\;\; L^2(\R^{d+1}_{t,q}).
\]
Now it remains to remove the matrix $\Pi^\pm(q,\delta)$ from the right hand side of the equation
\[
\op_\eps\!\left(\chi a\Pi^\pm\right) \psi^\eps_t = \op_\eps\!\left(\chi a\theta_{\eps,R}^\pm\Pi^\pm\right) \psi^\eps_t+ O(R^{-2}) + O(R^{-1}\sqrt\eps).
\]
In view of 
 \begin{equation*}\label{eq2}
 \chi a \theta^\pm_{\eps,R} = \chi a \theta^\pm_{\eps,R} \Pi^\pm + \chi a \theta^\pm_{\eps,R} \Pi^\mp,
 \end{equation*} 
we only need to prove that 
\begin{equation}\label{eq:small}
\op_\eps( \chi a \theta^\pm_{\eps,R} \Pi^\mp)\psi^\eps_t = O(R^{-2}) + O(R^{-1}\sqrt\eps)\quad\text{in}\;\; L^2(\R^{d+1}_{t,q}).
\end{equation}
We observe that, for $\eps$ small enough, 
$$
\theta^\pm_{\eps,R}=\theta^\pm_{\eps,R}(1-\theta^\mp_{\eps,R}),
$$
since 
$$
\left| \tau+\tfrac12|p|^2 + \alpha(q) \pm \tfrac12 g(q,\delta)\right| \le R\sqrt\eps
$$
on the support of $\theta^\pm_{\eps,R}$. Using again the Schr\"odinger equation, symbolic calculus and the estimate~(\ref{eq:Pi}), we obtain the desired relation \eqref{eq:small}. 
\end{proof}

We now reconsider the Landau--Zener transformation of Section~\ref{sec:reduction} in terms of expectation values. By Lemma~\ref{lem:loc},
$$
\left\langle\op_\eps(\chi a \Pi^\pm)\psi^\eps_t,\psi^\eps_t\right\rangle = \big\langle\op_\eps(\chi a \theta^\pm_{\eps,R})\psi^\eps_t,\psi^\eps_t\big\rangle + O(R^{-2}) + O(R^{-1}\sqrt\eps).
$$
Next, we rewrite the expectation value using the function $v^\eps = K^*_\eps\op_\eps(B^*_\eps)^{-1}\psi^\eps_t$ that has been introduced in Proposition~\ref{prop:quant}. Since $B_\eps = B + \eps B_1$ with $B^*B = \lambda$,  symbolic calculus implies
\begin{align*}
&\left\langle\op_\eps(\chi a \Pi^\pm)\psi^\eps_t,\psi^\eps_t\right\rangle\\
& \quad= 
\left\langle\op_\eps(\chi a \lambda \theta^\pm_{\eps,R})\op_\eps(B^*_\eps)^{-1}\psi^\eps_t,\op_\eps(B^*_\eps)^{-1}\psi^\eps_t\right\rangle + O(R^{-2}) + O(R^{-1}\sqrt\eps).
\end{align*}
In the presence of the symbol $\theta^\pm_{\eps,R}$ that loses a factor $R\sqrt\eps$ per derivative, the application of the Fourier integral operator $K^\eps$ yields
\begin{align*}
&\left\langle\op_\eps(\chi a \Pi^\pm)\psi^\eps_t,\psi^\eps_t\right\rangle\\ 
&\quad= 
\big\langle  \op_\eps\big((\chi a \lambda \theta^\pm_{\eps,R})\circ\kappa_\delta\big)  v^\eps, v^\eps\big\rangle + O(R^{-2}) + O(R^{-1}\sqrt\eps).
\end{align*}
In the next step, we move towards $\widetilde v^\eps = \op_\eps(M^\eps_0+\delta M^\eps_1)v^\eps$. Since the matrix $M_0^\eps$ is unitary and 
$|\delta|\le R\sqrt\eps$, we have
\begin{align*}
&\left\langle\op_\eps(\chi a \Pi^\pm)\psi^\eps_t,\psi^\eps_t\right\rangle\\
&\quad = \big\langle  \op_\eps\big((\chi a \lambda \theta^\pm_{\eps,R})\circ\kappa_\delta\big)  \widetilde v^\eps, \widetilde v^\eps\big\rangle + O(R^{-2}) + O(R\sqrt\eps).
\end{align*}
We then arrive at the solution $\check v^\eps$ of the Landau--Zener system~\eqref{eq:systreduitbis} by
\begin{align}\label{eq:exp_LZ}
&\left\langle\op_\eps(\chi a\Pi^\pm)\psi^\eps_t,\psi^\eps_t\right\rangle \nonumber\\ 
&\quad= 
\big\langle \op_\eps((\chi a \lambda \theta^\pm_{\eps,R})\circ\kappa_\delta) \check v^\eps,\check v^\eps\big\rangle + O(R^{-2}) + O(R\sqrt\eps) + O(r^3\eps^{1/8}).
\end{align}
At this stage of the proof of Theorem~\ref{theorem}, we have rewritten the expectation values for the Schr\"odinger solution $\psi^\eps_t$ in terms of the 
Landau--Zener solution $\check v^\eps$. Next, we reformulate the Markov process in the new coordinates. 


\subsection{The Markov process for the normal form}\label{sec:Markov}
 
Let us now introduce a new Markov process for effectively describing the dynamics of the reduced Landau--Zener problem \eqref{eq:systreduitbis}. We use the analogous building blocks as for the original process that defines the semigroup~$({\mathcal L}_{\eps}^t)_{t\ge0}$. 
We shall prove that this new process is close to the image of the original Markov process by the canonical transform $\kappa_\delta$.

\subsubsection{The image of the classical trajectories by the canonical transform}
We observe that by the transformation in~(\ref{def:B}), the eigenvalues of the reduced system~(\ref{eq:systreduit}) satisfy 
$$
-\sigma + j \, \sqrt{s^2+\delta^2\check\gamma(s,z,\sigma,\zeta)^2} = \left(\lambda \left(\tau+\Lambda^\mp(q,p,\delta)\right)\right)\circ\kappa_\delta,
$$
where the sign $j=\pm1$ depends on  the numbering discussed below in \S\ref{sec:numbering}. As a consequence, 
the integral curves of 
$$
H_{-\sigma \pm\sqrt{s^2+\delta^2\check\gamma(s,z,\sigma,\zeta)^2}}
$$ 
are mapped by the canonical transform $\kappa_\delta$ to those of 
$$
H_{\lambda\left(\tau+\Lambda^\mp(q,p,\delta)\right)}=\lambda H_{\tau+\Lambda^\mp(q,p,\delta)}+\left(\tau+\Lambda^\mp(q,p,\delta)\right)H_\lambda.
$$
Since the energy $\tau+\Lambda^\mp(q,p,\delta)$ is of order $R\sqrt \eps$ in our zone of observation, these trajectories are those of $\lambda H_{\tau+\Lambda^\mp(q,p,\delta)}$ up to some term of order $R\sqrt\eps$. Since the function $\lambda$ does not vanish in our zone of observation, the integral curves of $\lambda H_{\tau+\Lambda^\mp(q,p,\delta)}$ are those of $ H_{\tau+\Lambda^\mp(q,p,\delta)}$ up to a change of parametrization of the curve. 
We therefore consider the image of the Hamiltonian curves of the initial Markov process by the canonical transform $\kappa_\delta$ as being close to those of the functions 
$$
\widetilde\Lambda^\pm(s,z,\sigma,\zeta) = -\sigma\mp\sqrt{s^2+\delta^2\check \gamma(s,z,\sigma,\zeta)^2.}
$$


\subsubsection{The relation of the energies}
For discussing the relation of the different energies occuring in our analysis, we also introduce the functions 
\begin{equation}\label{eigenvaluesbis}
\widetilde\Lambda^\pm_0(s,z,\sigma,\zeta) = -\sigma\mp\sqrt{s^2+\delta^2\check \gamma_0(z,\zeta)^2}
\end{equation}
that belong to the Landau--Zener system \eqref{eq:systreduitbis}. We observe that both energies $\widetilde\Lambda^\pm$ and $\widetilde\Lambda^\pm_0$ satisfy
$$
\widetilde\Lambda^\pm, \widetilde\Lambda^\pm_0 = -\sigma \mp |s| + O(\delta^2|s|^{-1}),
$$
where we have used 
$$
\widetilde\Lambda^\pm_0(s,z,\sigma,\zeta) \pm |s|= \mp\frac{\delta^2\check \gamma_0(z,\zeta)^2}{|s|+\sqrt{s^2+\delta^2\check \gamma_0(z,\zeta)^2}}
$$
and a similar relation for $\widetilde\Lambda^\pm$. Therefore, since $\delta\le R\sqrt\eps$, we obtain
$$
\widetilde\Lambda^\pm - \widetilde\Lambda_0^\pm = O(R^2\eps|s|^{-1})
$$
and for any smooth cut-off function $\widetilde\theta\in{\mathcal C}^\infty_c(\R)$
$$
\widetilde\theta\!\left(\frac{\widetilde\Lambda^\pm(s,z,\sigma,\zeta)}{R\sqrt\eps}\right) = \widetilde\theta\!\left(\frac{\widetilde\Lambda_0^\pm(s,z,\sigma,\zeta)}{R\sqrt\eps}\right) + O(R\sqrt\eps|s|^{-1}).
$$
Hence a change in the energy localisation from $\widetilde\Lambda^\pm$ to $\widetilde\Lambda^\pm_0$ causes a deviation of the order $O(R^{-1}r^{-1})= O(R^{-1})$, when choosing $s\sim r R^2\sqrt\eps$ provided $r \gg 1$.


\subsubsection{The numbering of the eigenvalues}\label{sec:numbering}
The classical trajectories of the Landau--Zener system \eqref{eq:systreduitbis} are generated by the eigenvalues
$$
-\sigma \pm\sqrt{s^2 + \delta^2 \check\gamma_0(z,\zeta)^2}.
$$
For enumerating these eigenvalues such that the classical trajectories in in the original and the new coordinates can be naturally linked, we consider the case $\delta =0$ and use the vectors $H$ and $H'$ that have been defined in~(\ref{def:H}) and~(\ref{def:H'}). 

\medskip
We recall that $H$ and $H'$ are associated with ingoing trajectories for $\Lambda^+$ and $\Lambda^-$, respectively, and satisfy $\omega(H,H')<0$. Up to some perturbation term of order $R\sqrt\eps$, the canonical transformation $\kappa_0^{-1}$ sends $H$ and $H'$ on vectors that are collinear to $-\partial_s-\partial_\sigma$ and $-\partial_s+\partial_\sigma$ above the singular set $\{s=0\}$. Since
$$
\omega(-\partial_s-\partial_\sigma,-\partial_s+\partial_\sigma)= 2>0,
$$
the vector $H$ is collinear to $-\partial_s+\partial_\sigma$ and $H'$ to $-\partial_s-\partial_\sigma$ above $\{s=0\}$. Since
$$
-\partial_s+\partial_\sigma=H_{-\sigma-|s|},\; -\partial_s-\partial_\sigma=H_{-\sigma+|s|}\quad\text{on}\quad\{s>0\}
$$
and $s>0$ on ingoing trajectories, the vector field $H_{-\sigma-|s|}$ corresponds to the plus mode, while $H_{-\sigma+|s|}$ belongs to the minus mode. We 
therefore number the eigenvalues in the new coordinates according to \eqref{eigenvaluesbis}.


\subsubsection{The Hamiltonian trajectories}
The eigenvalues $\widetilde\Lambda^\pm_0$ generate the Hamiltonian systems
$$
\dot s = -1,\quad \dot z = \mp \frac{\delta^2 \check\gamma_0\partial_\zeta \check\gamma_0}{\sqrt{s^2+\delta^2\check \gamma_0^2}},\quad
\dot\sigma = \pm \frac{s}{\sqrt{s^2+\delta^2\check \gamma_0^2}},\quad \dot\zeta = \pm \frac{\delta^2 \check\gamma_0\partial_z \check\gamma_0}{\sqrt{s^2+\delta^2\check \gamma_0^2}},
$$
with corresponding flow maps 
$$
\widetilde\Phi_{0,\pm}^{\beth}:\R^{2d+2}\to\R^{2d+2}.
$$ 
Using that $|\dot z| + |\dot\zeta| = O(\delta|s|^{-1})$ and in view of \eqref{eigenvaluesbis}, we have
\begin{equation}\label{eq:flow}
\widetilde\Phi_{0,\pm}^{\beth}(s,z,\sigma,\zeta) = (s-\beth,z,\mp|s-\beth|,\zeta) + O(\delta |s|^{-1})
\end{equation}
for all points $(s,z,\sigma,\zeta)$ in our zone of observation and propagation times $\beth>0$. 
Note that for $s\sim r R^2\sqrt\eps$, we obtain an error of the order $R^{-1}r^{-1}$, which is smaller than $R^{-1}$.


\subsubsection{The non-adiabatic transitions}
Monitoring the gap function 
$$
\widetilde g_0(s,z,\zeta) = 2\sqrt{s^2+\delta^2\check \gamma_0(z,\zeta)^2}
$$
along the Hamiltonian trajectories associated with $\widetilde\Lambda^\pm_0$ we look for points in $\R^{2d+2}$, where a local minimum is attained. We obtain the condition 
$$
0 = \partial_s \widetilde g_0 \;\partial_\sigma\widetilde\Lambda^\pm_0 + \nabla_z\widetilde g_0 \cdot \nabla_\zeta \widetilde\Lambda^\pm_0 
- \partial_\sigma\widetilde g_0\; \partial_s \widetilde \Lambda^\pm_0 - \nabla_\zeta \widetilde g_0\cdot \nabla_z \widetilde\Lambda^\pm_0 = -\partial_s \widetilde g_0,
$$
that is equivalent to $s=0$.  
Hence, the new jump manifold is the set
$$
\widetilde\Sigma_\eps = \left\{(s,z,\sigma,\zeta)\in\R^{2d+2}\mid s=0,\; 2\delta|\check\gamma_0(z,\zeta)|\le R\sqrt\eps\right\}.
$$
The Landau--Zener formula of \cite[Proposition~7]{FG03}, see \S\ref{sec:LZ_form}, suggests to perform non-adiabatic transitions with probability
$$
\widetilde T_\eps(z,\zeta) := \exp\!\left(-\frac{\pi}{\eps} \delta^2\check\gamma_0(z,\zeta)^2\right)
$$
when reaching the jump set $\widetilde\Sigma_\eps$. By the energy localization of our observables, we have $s=0$ and $\sigma=O(R\sqrt\eps)$ on the jump manifold $\widetilde\Sigma_\eps$. By Theorem~\ref{prop:parametrization}, a Taylor expansion around $\sigma=0$ reads
$$
\check\gamma(0,z,\sigma,\zeta) = \check\gamma_0(z,\zeta) + \sigma \partial_\sigma\check\gamma(0,z,0,\zeta) + O(\sigma^2/\delta_0).
$$
Using that $\delta/\delta_0$ is bounded, we obtain
$$
\frac{\delta^2}{\eps}\left(\check\gamma(0,z,\sigma,\zeta) - \check\gamma_0(z,\sigma)\right) = O(R^3\sqrt\eps)
$$
and
$$
\widetilde T_\eps(z,\zeta) = \exp\!\left(-\frac{\pi}{\eps}\delta^2\check\gamma(s,z,\sigma,\zeta)^2\right) + O(R^3\sqrt\eps).
$$
Lemma~\ref{lem:gamma} then provides
\begin{equation}\label{eq:LZ}
\widetilde T_\eps = T_\eps\circ\kappa_\delta + O(R^3\sqrt\eps),
\end{equation}
where the original transition rate $T_\eps(q,p,\delta)$ has been defined in \eqref{Tepsdelta}. Besides, the original jump condition $p\cdot\nabla_q g(q,\delta) = 0$ is equivalent to
$$
\widetilde\beta(q,\delta) (p\cdot\nabla_q\widetilde\beta(q,\delta)) + \delta^2 \widetilde\gamma(q) (p\cdot\nabla_q\widetilde\gamma(q)) = 0, 
$$
that is, $s=O(\delta^2)$. Hence, the non-adiabatic transitions of the new and the original Markov process mostly differ due to the transition rate estimate \eqref{eq:LZ}. 

\subsubsection{The drift}\label{subsubsec:drift}
At $s=0$,  we observe for all $\delta\ge0$  the energy relation
\begin{equation}\label{driftunicity}
\widetilde\Lambda^\pm_0(0,z,\sigma\mp2 \delta |\check \gamma_0(z,\zeta)|,\zeta)=\widetilde\Lambda^\mp_0(0,z,\sigma,\zeta).
\end{equation}
The best drift is, of course, the exact one and is given by
$$
\widetilde J_\pm:\sigma\mapsto \sigma\mp2 \delta |\check \gamma_0(z,\zeta) |.
$$ 
This exact drift is performed in the direction of~$\partial_\sigma$ that is collinear to the difference of the two Hamiltonian vector fields in the particular case $\delta=0$, motivating the geometric underpinning of the original drift construction, see Remark~\ref{rem:energy}. 
We also note  that the size of the drift $2\delta|\check\gamma_0(z,\zeta)|$ is precisely the gap size for points in the jump manifold.

\subsection{The semigroup for the normal form}\label{sec:semigroup}

The Markov process described in the previous section \S\ref{sec:Markov} defines a semigroup $\widetilde{\mathcal L}_\eps$ acting on functions in the space 
$$
\widetilde{\mathcal B} = \left\{f:\R^{2d+2}\times\{-1,1\}\to\C\mid f\;\text{is measurable, bounded}\right\}.
$$ 
Following the normal form transformation of the expectation values given in \eqref{eq:exp_LZ}, we consider a symbol $c_{\eps,R}^{out}\in\widetilde{\mathcal B}$ whose plus-minus-components are defined by  
$$
c_{\eps,R}^{\pm,out}(s,z,\sigma,\zeta)=   b^\pm (s,z,\sigma,\zeta)\;
\widetilde\theta\!\left(\frac{\widetilde\Lambda^\pm_0(s,z,\sigma,\zeta)}{R\sqrt\eps}\right),
$$
where the functions $b^\pm$ and $\widetilde\theta$ have the following properties: $b^\pm(s,z,\sigma,\zeta)$ are two smooth functions compactly supported in the outgoing region $\{s<0\}$ such that the random trajectories reaching their support have only one transition during the observation time that is of length $\beth>0$. The functions $b^\pm(s+\beth,z,\sigma,\zeta)$ are supported in the incoming region $\{s>0\}$. The cut-off function $\widetilde\theta\in{\mathcal C}^\infty_c(\R)$ satisfies
\begin{align*}
0\leq \widetilde\theta\leq 1,\quad 
\widetilde \theta(u)=0\;\text{for}\;|u|>1,\quad
\widetilde \theta(u)=1\;\text{for}\;|u|<1/2.
\end{align*}
We now analyse the pull back of the symbol $c_{\eps,R}^{out}\in\widetilde{\mathcal B}$ by the semigroup for a suitably chosen time $\beth>0$,  
$$
c_{\eps,R}^{in}:=\widetilde{\mathcal L}_{\eps}^\beth\, c_{\eps,R}^{out}.
$$
By Assumption~$(A2)_\delta$ of \S\ref{sec:delta}, we only consider transitions generated by one incoming mode. We assume that it is the plus mode and denote the two leading order contributions of $c^{+,in}_{\eps,R}$ by $c^{+,in}_{\eps,R,+}$ and $c^{+,in}_{\eps,R,-}$, where the subscript depends on the outgoing mode. Our aim is to relate
$$
\left\langle\op_\eps(c_{\eps,R}^{\pm,out})\check v^\eps,\check v^\eps\right\rangle \qquad\text{with}\qquad
\left\langle\op_\eps(c_{\eps,R,\pm}^{+,in})\check v^\eps,\check v^\eps\right\rangle.
$$
If the incoming trajectories are associated with the minus mode, the arguments are analogous.

 
\subsubsection{Transport without transitions} 
The component $c_{\eps,R,+}^{+,in}$ takes into account  the classical transport along the plus trajectories
and the probability of staying on the same mode.
By conservation of energy along classical trajectories, we have   
$$
\widetilde\Lambda^+_0(\widetilde \Phi^{-\beth}_{0,+}(s-\beth,z,\sigma,\zeta))= \widetilde\Lambda^+_0(s-\beth,z,\sigma,\zeta).
$$
By \eqref{eq:flow}, we therefore deduce
\begin{align*}
& c_{\eps,R,+}^{+,in}(s-\beth,z,\sigma,\zeta) = \left(1- \widetilde T _\eps(z,\zeta)\right) 
(c_{\eps,R}^{+,out}\circ\widetilde \Phi^{-\beth}_{0,+})(s-\beth,z,\sigma,\zeta)\\
& =   \left(1- \widetilde T _\eps(z,\zeta)\right) b^+(s,z,-|s|,\zeta) 
\;\widetilde\theta\!\left(\frac{\widetilde\Lambda^+_0(s-\beth,z,\sigma,\zeta)}{R\sqrt\eps}\right) +O(R\sqrt\eps),
\end{align*}
so that
\begin{align*}
& c_{\eps,R,+}^{+,in} (s,z,\sigma,\zeta)\\
&  = \left(1- \widetilde T _\eps(z,\zeta)\right)  b^+(s+\beth,z,-|s+\beth|,\zeta)\;
\widetilde \theta\!\left(\frac{\widetilde\Lambda^+_0(s,z,\sigma,\zeta)}{R\sqrt\eps}\right) +O(R\sqrt\eps).
 \end{align*}
 

\subsubsection{Transport and transitions with drift} 
The component  $c_{\eps,R,-}^{+,in}$ is more intricate, since it incorporates classical transport through both modes, application of the transfer coefficient and of the drift. Indeed, the branches of minus trajectories which reach the support of $c_{\eps,R}^{-,out}$ result from plus trajectories that have been drifted.  More precisely, we have 
$$
(\widetilde \Phi^{s}_{0,-}\circ\widetilde J_+\circ\widetilde\Phi^{-s-\beth}_{0,+})(s-\beth,z,\sigma,\zeta) = 
\left(s,z,|s|,\zeta\right) + O(R\sqrt\eps).
$$
for $(s,z,\sigma,\zeta)\in{\rm supp}(c_{\eps,R}^{-,out})$.  By conservation of energy and the drift relation~(\ref{driftunicity}) we have
$$
\widetilde \Lambda^-_0\!\left(\widetilde\Phi^s_{0,-}\circ \widetilde J_+ \circ \widetilde\Phi_{0,+}^{-s-\beth}(s-\beth,z,\sigma,\zeta)\right)=
\widetilde \Lambda^+_0(s-\beth,z,\sigma,\zeta).
$$
Applying the transfer coefficient, we obtain 
 \begin{eqnarray*}
c_{\eps,R,-}^{+,in}(s-\beth,z,\sigma,\zeta)& = &  \widetilde T_\eps(z,\zeta)\,
(c_{\eps,R}^{-,out}\circ\widetilde \Phi^{s}_{0,-}\circ\widetilde J_+\circ\widetilde\Phi^{-s-\beth}_{0,+})(s-\beth,z,\sigma,\zeta)\\
& = &  \widetilde T_\eps(z,\zeta)\, b^-(s,z,|s|,\zeta)\,  \widetilde\theta\!\left(\frac{\widetilde\Lambda^+_0(s-\beth,z,\sigma,\zeta)}{R\sqrt\eps}\right) +O(R\sqrt\eps),
\end{eqnarray*}
that is,
\begin{equation*}
c_{\eps,R,-}^{+,in} (s,z,\sigma,\zeta)  =  \widetilde T_\eps(z,\zeta)   \,b^-(s+\beth, z,|s+\beth|,\zeta)\,
 \widetilde\theta\!\left(\frac{\widetilde\Lambda^+_0(s,z,\sigma,\zeta)}{R\sqrt\eps}\right)+O(R\sqrt\eps).
 \end{equation*}
 
\subsection{The transitions}\label{sec:transitions}

We now prove that the semigroup $\widetilde{\mathcal L}_\eps$ effectively describes the dynamics of the Landau--Zener system \eqref{eq:systreduitbis} in the sense that  
\begin{equation}\label{claim}
\left\langle\op_\eps(c_{\eps,R}^{\pm,out})\check v^\eps,\check v^\eps\right\rangle =
\left\langle\op_\eps(c_{\eps,R,\pm}^{+,in})\check v^\eps,\check v^\eps\right\rangle + O(\eta_\eps),
\end{equation}
where the error term is obtained as
$$
O(\eta_\eps) = O(r^{-1}) + O(R^{-1}) + O(\eps R^2 \ln(r R)).
$$
By Theorem~\ref{prop:parametrization} the function $\check\gamma_0$ is a smooth function with bounded derivatives. Thus, we can 
follow the argumentation developed in~\cite[\S5.3]{FL08}, crucially using the operator-valued Landau--Zener formula of~\cite[Proposition~7]{FG03}.

\subsubsection{Using energy localization}
We work with the eigenprojectors
$$
\widetilde\Pi^\pm_0(s,z,\zeta) = \frac12\left( \Id \mp \frac{1}{\sqrt{s^2+\delta^2\check\gamma_0(z,\zeta)^2}}  
\begin{pmatrix}s & \delta\check\gamma_0(z,\zeta) \\ \delta\check\gamma_0(z,\zeta)& -s \end{pmatrix}\right)
$$
of the Landau--Zener system, numbered consistently with the eigenvalues in \eqref{eigenvaluesbis}. We observe that for $|s| \sim r R^2\sqrt\eps$
\begin{eqnarray}
\nonumber
\widetilde\Pi^+_0(s,z,\zeta)&=&
\begin{pmatrix}0&0\\ 0&1\end{pmatrix}+O(R^{-1})
\;\;{\rm in}\;\;\{s>0\},\\
\label{eq:Pitilde}
\widetilde\Pi^+_0(s,z,\zeta)&=&\begin{pmatrix}1&0\\0&0\end{pmatrix}+O(R^{-1})
\;\;{\rm in}\;\;\{s<0\},
\end{eqnarray}
and obtain similar asymptotics for $\widetilde\Pi^-_0$ since ${\rm Id}=\widetilde \Pi^+_0+\widetilde \Pi^-_0$. 
By Lemma~\ref{lem:loc}, we then have 
$$
\left\langle\op_\eps( c_{\eps,R}^{\pm,out}) \check v^\eps,\check v^\eps\right\rangle  =  
\left\langle\op_\eps(b^{\pm}\widetilde\Pi^\pm_0)\check v^\eps,\check v^\eps\right\rangle+O(R^{-2}) + O(R^{-1}\sqrt\eps), 
$$
and consequently, 
\begin{eqnarray*}
\left\langle\op_\eps( c_{\eps,R}^{+,out})\check v^\eps,\check v^\eps\right\rangle &=&\Big\langle\op_\eps(b^{+,out})\check v^\eps_1,\check v^\eps_1\Big\rangle+
O(R^{-1}),\\
\left\langle\op_\eps( c_{\eps,R}^{-,out})\check v^\eps,\check v^\eps\right\rangle &=&\Big\langle\op_\eps( b^{-,out}) \check v^\eps_2,\check v^\eps_2\Big\rangle+
O(R^{-1})
\end{eqnarray*}
with
$$
b^{\pm,out}(s,z,\sigma,\zeta)= b^\pm (s,z,\sigma,\zeta)
$$
supported in the outgoing region $\{-r R^2\sqrt\eps \le s\le -{r \over 2}  R^2\sqrt\eps\}$. Similarly, we have 
\begin{eqnarray*}
\left\langle\op_\eps(c_{\eps,R,+}^{+,in}) \check v^\eps,\check v^\eps\right\rangle &=& 
\left\langle\op_\eps\!\left((1-\widetilde T_\eps)  b^{+,in}\right)\check v^\eps_2,\check v^\eps_2\right\rangle+O(R^{-2}) + O(R^{-1}\sqrt\eps),\\
\left\langle\op_\eps(c_{\eps,R,-}^{+,in}) \check v^\eps,\check v^\eps\right\rangle &=&
\left\langle\op_\eps\!\left(\widetilde T_\eps \,b^{-,in}\right)\check v^\eps_2,\check v^\eps_2\right\rangle+O(R^{-2}) + O(R^{-1}\sqrt\eps)
\end{eqnarray*}
with 
$$
b^{\pm,in}(s,z,\sigma,\zeta)  =  b^\pm(s+\beth,z,\mp|s+\beth|,\zeta).
$$
supported in the incoming region $\{{r \over 2} R^2\sqrt\eps \le s\le r R^2\sqrt\eps\}$.


\subsubsection{The Landau--Zener formula}\label{sec:LZ_form}
Following \cite[Proposition~7]{FG03}, we rewrite the Landau--Zener system~\eqref{eq:systreduitbis} as
\begin{equation}\label{eq:LZ_G}
\frac{\eps}{i} \partial_s \check v^\eps = \begin{pmatrix}s & \sqrt\eps G\\ \sqrt\eps G^* & -s\end{pmatrix} \check v^\eps
\quad\text{with}\quad G = \frac{\delta}{\sqrt\eps}\op_\eps(\check\gamma_0(z,\zeta)).
\end{equation}
Then there exist two vector-valued functions 
$k^{\eps,\pm}\in L^2(\R^d,\C^2)$ such that  for any cut-off function $\chi\in{\mathcal C}_c^\infty([0,R^2])$ and for $\pm s>0$
\begin{eqnarray*}
\chi(GG^*) \check v^\eps_1(z,s) & = & 
\chi(GG^*) {\rm e}^{is^{2}/(2\eps)} 
\left|{\tfrac{s}{\sqrt \eps}}\right|^{i\frac{GG^*}{2}}
k^{\eps,\pm}_{1}(z)+O(R^2\sqrt\eps/s),\\
\chi(G^*G)\check v^\eps_2(z,s) & = & 
\chi(G^*G){\rm e}^{-is^{2}/(2\eps)}
\left|{\tfrac{s}{\sqrt \eps}}\right|^{-i\frac{G^*G}{2}}
k^{\eps,\pm}_{2}(z)+O(R^2\sqrt\eps/s),
\end{eqnarray*}
where $k^{\eps,+}=S_\eps k^{\eps,-}$ with
$$
S_\eps=
\begin{pmatrix}a(GG^*) & -
 \overline b(GG^*)G \\  b(G^*G)G^* & a(G^*G)\end{pmatrix}.
$$ 
The functions defining the scattering matrix satisfy
$$
a(\lambda)={\rm e}^{-\pi\lambda/2},\qquad a(\lambda^2) + \lambda |b(\lambda)|^2 = 1,\qquad \lambda\in\R.
$$
Moreover, the asymptotics of \cite[Lemma 8 \& 9]{FG03} provide for any smooth and compactly supported symbol $\phi\in{\mathcal C}^\infty_c(\R^{2d+2})$
\begin{equation}\label{eq:phase}
\left|{\tfrac{s}{\sqrt\eps}}\right|^{\pm
i\frac{G^*G}{2 }} 
\op_\eps(\phi)
\left|{\tfrac{s}{\sqrt\eps}}\right|^{\mp i\frac{G^*G}{
2}}=
\op_\eps(\phi)+O(R^2\eps |\ln(s/\sqrt\eps)|).
\end{equation}
These asymptotics yield an error of the order $1/r$, which motivates to choose 
\[
r=r(\eps) = \eps^{-1/32},
\] 
so that the error $r^3\eps ^{1/8}$ in equations~(\ref{eq:error}) and \eqref{eq:exp_LZ} is of the same size as $1/r$.
\begin{remark}\label{rem:1/8ter}
If $\widetilde \gamma(\cdot,\delta)$ had uniformly bounded derivatives, then, in view of Remarks~\ref{rem:1/8} and~\ref{rem:1/8bis}, we could choose $r =R$ and  obtain an overall remainder of the order $1/R =\eps^{1/8}$.
\end{remark}


\subsubsection{Applying the Landau--Zener formula}
Using the Landau--Zener formalism described in \S\ref{sec:LZ_form}, we now restrict ourselves to proving
$$
\Big\langle\op_\eps(b^{+,out})\check v^\eps_1,\check v^\eps_1\Big\rangle = \left\langle\op_\eps\!\left((1-\widetilde T_\eps)  b^{+,in}\right)\check v^\eps_2,\check v^\eps_2\right\rangle+O(\eta_\eps),
$$
since the proof for the second estimate in \eqref{claim} is analogous. 

\medskip
We first use the relation between $\check v^\eps_1$ and $k^{\eps,-}_1$ on the outgoing region $\{s<0\}$ for $s\le-rR^2\sqrt\eps$ to obtain
$$
\left\langle\op_\eps(b^{+,out})\check v^\eps_1,\check v^\eps_1\right\rangle 
= \left\langle\op_\eps(b^{+,out})k^{\eps,-}_{1},k^{\eps,-}_{1}\right\rangle + O(r^{-1}) + O(\eps R^2 \ln(rR)).
$$
Then, we perform the change of variable $s\mapsto s-\beth$, 
$$
\left\langle\op_\eps(b^{+,out})\check v^\eps_1,\check v^\eps_1\right\rangle  
= \left\langle\op_\eps(b^{+,in})k^{\eps,-}_{1},k^{\eps,-}_{1}\right\rangle + O(r^{-1}) + O(\eps R^2\ln(rR)).
$$
Since we have assumed that the incoming minus contributions are negligibly small, we neglect the scattering contribution from $\check v^\eps_1$ and consequently from $k^{\eps,+}_1$. We therefore deduce from the scattering relation $k^{\eps,-} = S_\eps^* k^{\eps,+}$ that 
$$
k^{\eps,-}_1 = - G\,\overline b(GG^*) k^{\eps,+}_2 + O(\eta_\eps)
$$ 
and
\begin{align*}
\Big\langle\op_\eps(b^{+,out})\check v^\eps_1,\check v^\eps_1\Big\rangle  
&= 
\left\langle\op_\eps\!\left((1-\widetilde T_\eps)b^{+,in}\right)k^{\eps,+}_{2},k^{\eps,+}_{2}\right\rangle +O(\eta_\eps)\\
&=
\left\langle\op_\eps\!\left((1-\widetilde T_\eps)b^{+,in}\right)\check v^\eps_{2},\check v^\eps_{2}\right\rangle +O(\eta_\eps).
\end{align*}

%

\section{Numerical simulations}\label{sec:numerics}
 
We consider four specific examples of avoided crossings in one space dimension. The corresponding eigenvalues surfaces are plotted in Figure \ref{fig:surfaces}, 
while the detailed definition of the four model systems is given in Tables~\ref{tab:param1} and \ref{tab:param2}. Three examples are taken from Tully's 1994 paper \cite{Tu1} on the surface hopping algorithm of the fewest switches: the simple, the dual and the extended crossing. The arctangent crossing is included as an example, which meets the assumptions of our main Theorem~\ref{theorem}. Its eigenvalues are defined by smooth functions, and the coefficients except for the minimum gap parameter $\delta_0$ are of order one with respect to the semiclassical parameter~$\eps$. 
In all simulations, the initial data
$$
\psi_0(q) = (\pi\eps)^{-1/4} \exp\!\left(-{\textstyle\frac{1}{2\eps}(q-q_0)^2 + \frac{i}{\eps}} p_0(q-q_0)\right) e^\pm(q)
$$
are multiples of a Gaussian wave packet with phase space centers $(q_0,p_0)\in\R^2$ and a real-valued eigenvector $e^\pm(q)$ of the matrix $V(q)$. Following \cite{Tu1}, the three Tully examples have the  semiclassical parameter 
$$
\eps=1/\sqrt{2000}\approx 0.02,
$$ 
which roughly corresponds to the mass of the hydrogen atom of $1836$ atomic units. For the arc\-tangent crossing we have chosen $\eps=10^{-3}$. The time interval $[0,t_{\rm fin}]$ of all simulations allows that the wave packet passes the crossing region once. 

\begin{figure}
\includegraphics[width=\textwidth]{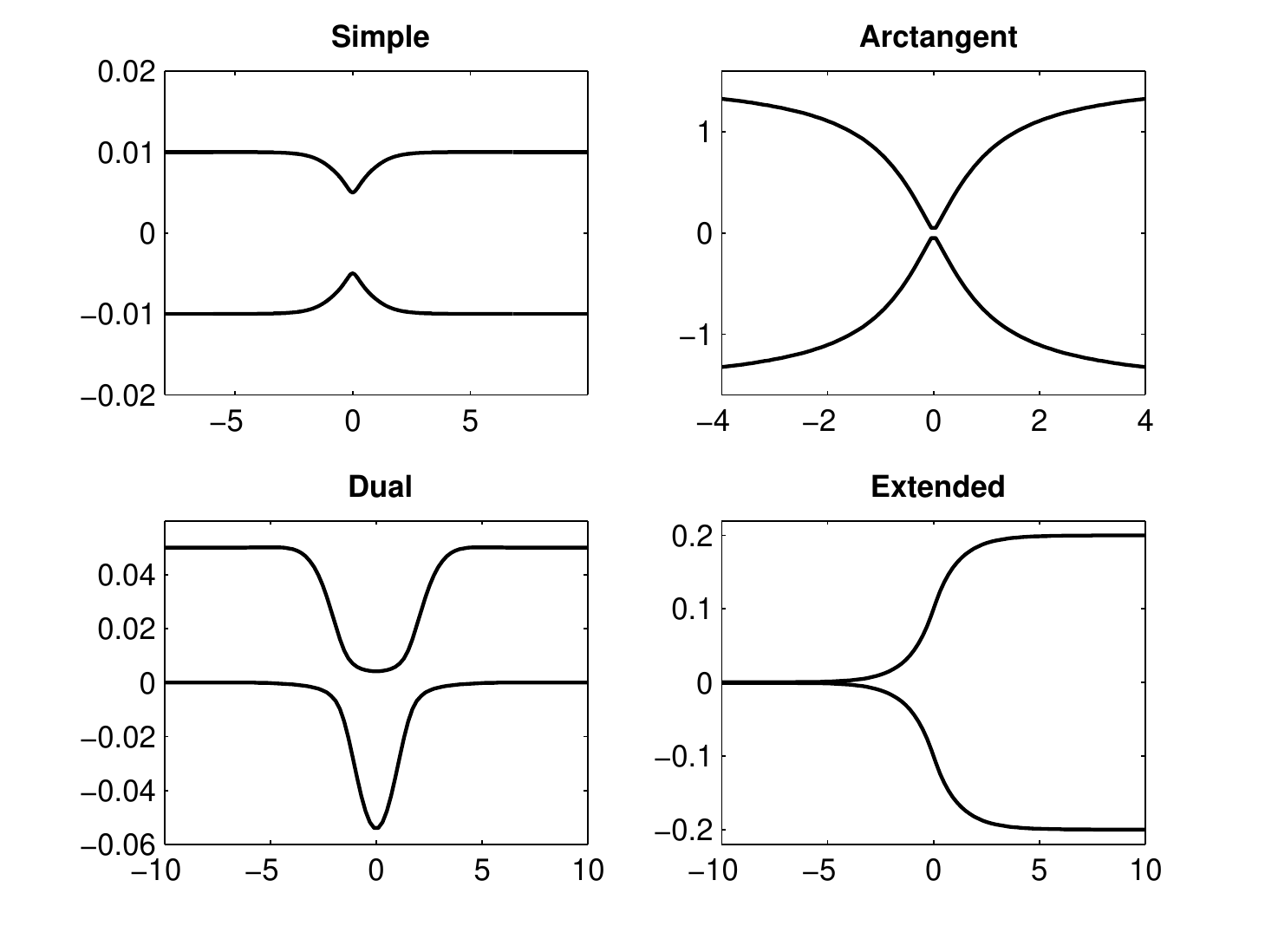}
\caption{The eigenvalue surfaces of the potentials considered for our numerical simulations. }
\label{fig:surfaces}
\end{figure}

\begin{table}[h]
\centering
\begin{tabular}{c|c|c}
& Simple & Arctangent \\\hline
$\eps$ & $2000^{-1/2}$ & $10^{-3}$ \\
$\delta_0$ & $0.005$ & $10^{-3/2}$ \\\hline
initial level & minus & plus \\
$(q_0,p_0)$ & $(-5,1)$ & $(-1,1)$ \\
$t_{\rm fin}$ & $10$ & $2$\\\hline
$\beta(q)$ & $0.01\,\sgn(q)\big(1-\e^{-1.6|q|}\big)$ & $\arctan(q)$ \\
$\gamma(q)$ & $\delta_0\e^{-q^2}$ & $\delta_0$ \\
$\alpha(q)$ & $0$ & $0$
\end{tabular}
\caption{\label{tab:param1}
Functions and parameters defining the simple and the arctangent crossings. In both cases, the eigenvalue surfaces 
have their minimal gap at $q=0$.}
\end{table}

\begin{table}[h]
\centering
\begin{tabular}{c|c|c}
& Dual & Extended\\\hline
$\eps$ & $2000^{-1/2}$ & $2000^{-1/2}$\\
$\delta_0$ & $0.015$ &  $6\cdot 10^{-4}$\\\hline
initial level & minus & plus\\
$(q_0,p_0)$ & $(-5,1)$ & $(0,-1)$\\
$t_{\rm fin}$ & $10$ & $10$\\\hline
$\beta(q)$ & $0.05 \e^{-0.28 q^2} - 0.025$ & $\delta_0$\\
$\gamma(q)$ & $\delta_0\e^{-0.06 q^2}$ & $0.1\sgn(q)(1-\e^{-0.9|q|})+0.1$\\
$\alpha(q)$ & $-\beta(q)$ & $0$\\
\end{tabular}
\caption{\label{tab:param2}
Functions and parameters defining the dual and the extended crossing. The dual crossing surfaces have their minimal gap at $q\approx\pm 1.6$. The surface gap of the extended crossing decreases monotonically as $q\to-\infty$.}
\end{table}


\subsection{A surface hopping algorithm}

Our analysis of the dynamics through an avoided eigenvalue crossing suggests 
a surface hopping algorithm formulated in terms of Wigner functions. Such an algorithm can either treat the effective 
Landau-Zener transitions by a deterministic branching scheme or by a probabilistic 
accept-reject mechanism. The probabilistic version, which will be discussed here, keeps the number of trajectories constant, 
which is to the best advantage for the memory requirements of the algorithm, see the simulations 
for a model of pyrazine and of the  ammonia cation \cite{LS,BDLT}.

\medskip

For notational simplicity, we restrict ourselves to the case, that the initial data are associated with the 
upper level. The same reasoning applies for initial data associated to the lower level with the obvious alterations. 
The probabilistic surface hopping algorithm works as follows: 

\subsubsection{Initial sampling}
Draw $N\in\N$ pseudorandom phase space samples 
$$
(q_1,p_1)^+,\ldots,(q_N,p_N)^+,
$$ 
which are independent and identically distributed according to $w^\eps_+(\psi_0)$. If the initial data are a Gaussian wave packet, then the Wigner function is given by the explicit formula  
$$
w^\eps_+(\psi_0)(q,p) = (\pi\eps)^{-1} \exp\!\left(-\tfrac{1}{\eps}|(q,p)-(q_0,p_0)|^2\right),
$$
that is the densitiy function of a bivariate normal distribution.

\subsubsection{Transport}
Propagate the sample points along the Hamiltonian curves of $\Phi^t_+$.

\subsubsection{Non-adiabatic transitions}
If a trajectory $t\mapsto(q_j^+(t),p_j^+(t))$ attains a local minimal gap of size smaller than $\sqrt\eps$ at time $t^*$ in the phase space point $(q^*,p^*)$, then draw a pseudrandom number $\zeta$ uniformly distributed in $[0,1]$. 
If $T_\eps(q^*,p^*)>\zeta$, then a hop occurs according to 
$$
(q^*,p^*,+) \longrightarrow (q^*,p_*+\omega(q^*,p^*),-).
$$
Otherwise, the trajectory continues on the upper level. 

\subsubsection{Computation of expectation values}
If at time $t$ there are $N^+$ trajectories on the upper and $N^-$ trajectories on the lower level, then the expectation values for observables 
$$
a(q,p)=a^+(q,p)\Pi^+(q)+a^-(q,p)\Pi^-(q)
$$ 
are approximated as 
\begin{equation}
\label{eq:exp}
(\op_\eps(a)\psi_t^\eps,\psi_t^\eps)\;\approx\; \frac{1}{N^+} \sum_{j=1}^{N^+}a^+(q_j^+(t),p_j^+(t)) + \frac{1}{N^-} \sum_{j=1}^{N^-}a^-(q_j^-(t),p_j^-(t)).
\end{equation}

\medskip
The overall accuracy of the approximation is then determined by the initial sampling, the discretization of the Hamiltonian flows, and the asymptotic accuracy of the surface hopping semigroup. If the flow discretization 
is a symplectic order $p$ method with time step $\Delta_t$, then the error of the approximation (\ref{eq:exp}) is
$$
O(1/\sqrt{N}) + O(\Delta_t^p) + O(\eps^\gamma),
$$
$\gamma=1/8$ according to Theorem~\ref{theorem}.
For the numerical experiments presented here, the initial sampling and the discretized classical flows have been accurate enough, such that the asymptotic $\eps$-dependent error of the algorithm is dominant. 

\medskip

Our reference values have been obtained from numerically converged solutions of the Schr\"odinger equation (\ref{eq:schro}), which have been computed by a Strang splitting scheme with Fourier collocation. 
All figures show the reference values as solid lines, while the little stars and circles mark values computed by the surface hopping algorithm. We note, that the space grid for the Fourier collocation must 
resolve the oscillations of the wave function, which is easily achieved in one space dimension. For higher dimensional problems, however, such discretizations suffer from the curse of dimensionality, 
which is not the case for our surface hopping algorithm.    


\subsection{The simple and the arctangent crossing}
\label{sec:simple}

\begin{figure}
\subfigure[Simple: $\eps\approx 0.02$, $\delta_0=0.005$]{\includegraphics[height=0.4\textheight]{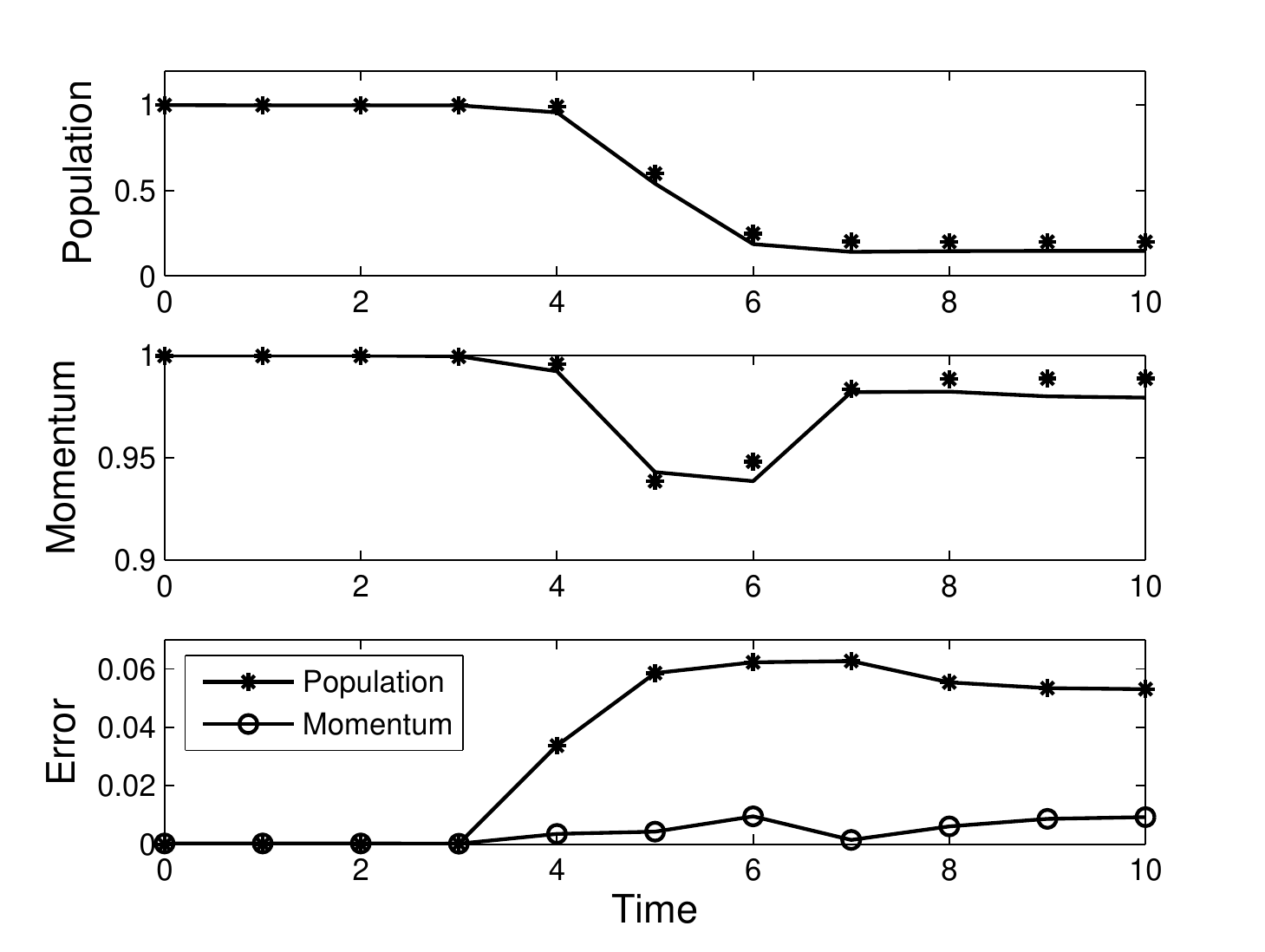}}
\subfigure[Arctangent: $\eps=10^{-3}$, $\delta_0 = \sqrt\eps$]{\includegraphics[height=0.4\textheight]{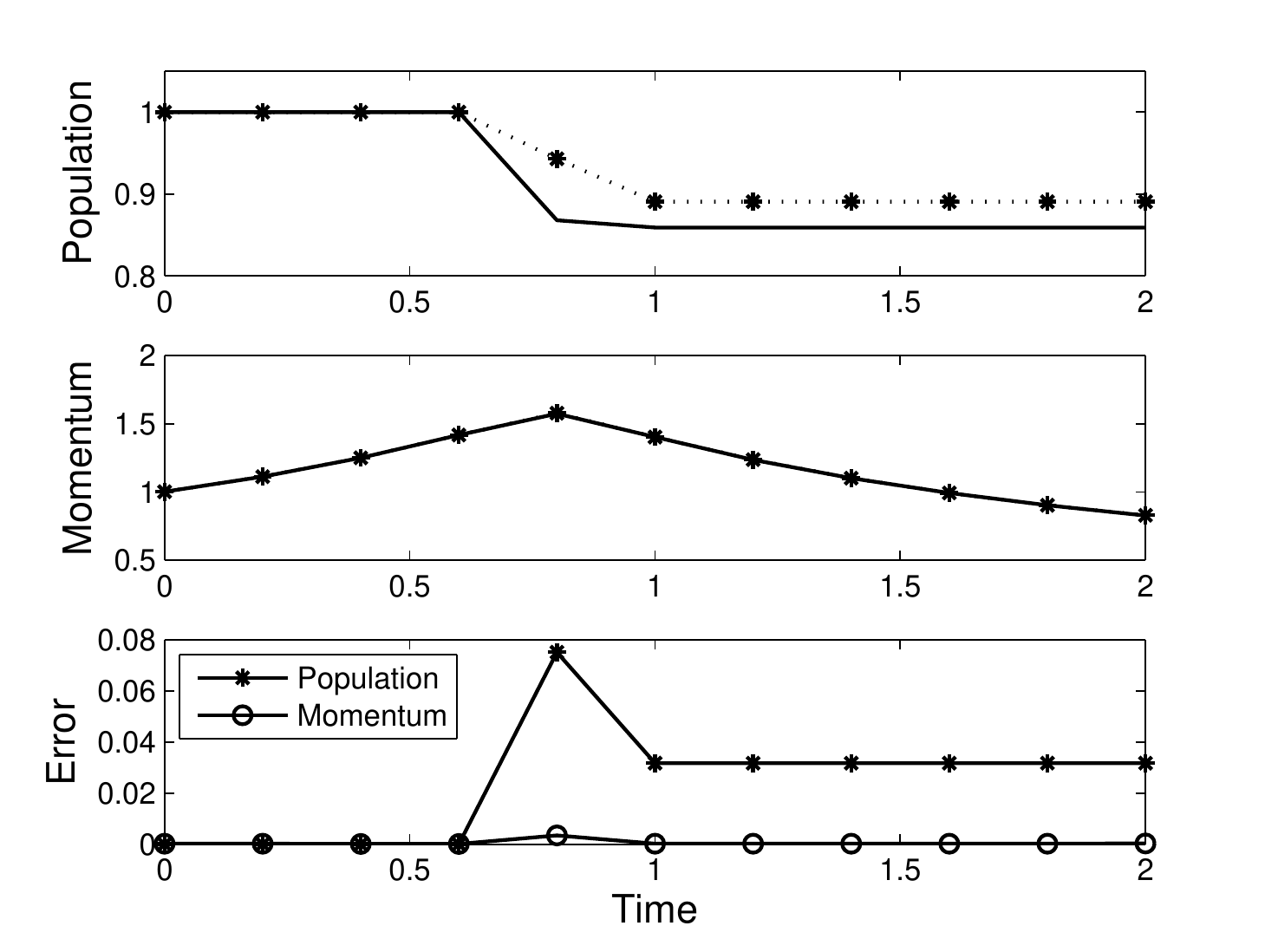}}
\caption{The simple and the arctangent crossing. The initial wave function is associated with the lower (a) and the upper level~(b). 
The results of the surface hopping algorithm, marked with stars and circles, are in good agreement with the reference.} 
\label{fig:simple}
\end{figure}

For both examples the surface hopping algorithm produces meaningful approximations of the dynamics even though the simple crossing has a non-smooth potential matrix and a surface gap just varying by a factor two.
Figure~\ref{fig:simple} shows population transfer away from the initial energy level and a corresponding change of the average momentum on the initial level. The final populations are approximated within an accuracy of 
$0.04$ to $0.05$. The error of the momentum expectation is even smaller.

\subsection{The dual and the extended crossing}
\label{sec:dual}

\begin{figure}
\subfigure[Dual: $\eps\approx 0.02$, $\delta_0=0.015$]{\includegraphics[height=0.4\textheight]{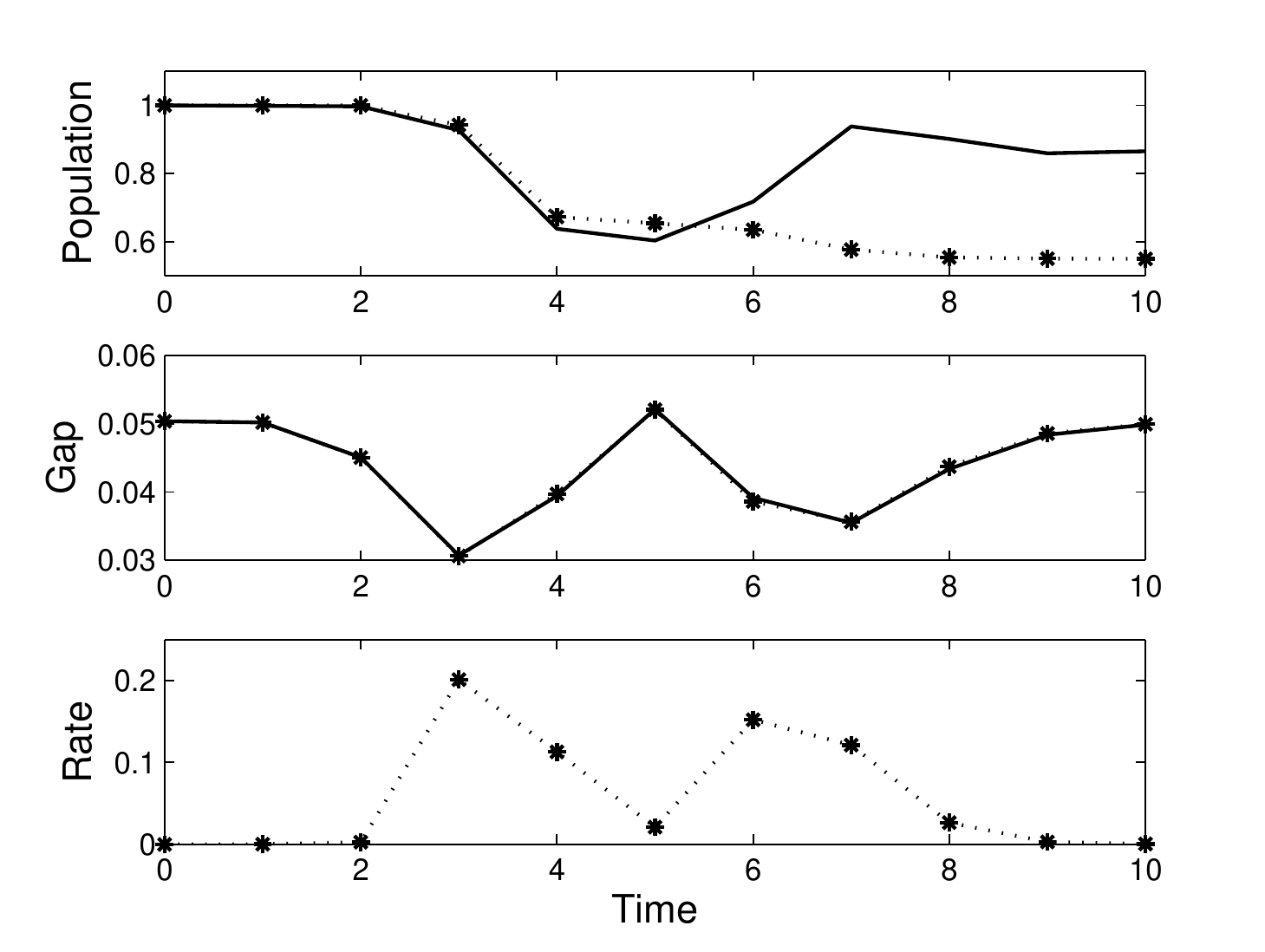}}
\subfigure[Extended: $\eps\approx 0.02$, $\delta_0=6\cdot10^{-4}$]{\includegraphics[height=0.4\textheight]{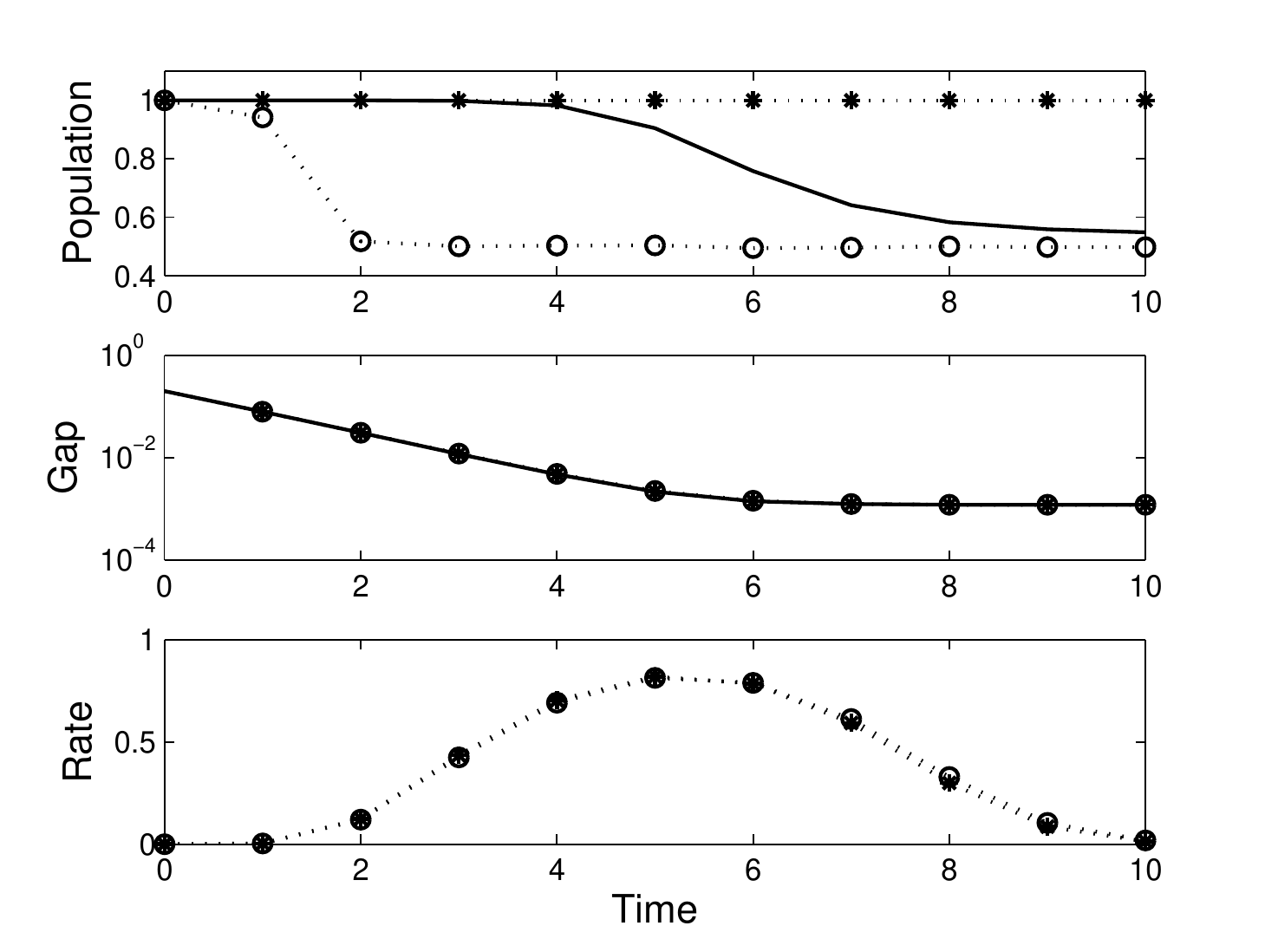}}
\caption{The dual and the extended crossing. The initial wave function is associated with the lower (a) and the upper eigenvector~(b). In both cases, the surface hopping algorithm expectedly fails to reproduce the dynamics.} 
\label{fig:dual}
\end{figure}

These two examples illustrate the limitations of our approximation. Since the dual crossing model has two subsequent crossings at $q\approx -1.6$ and $q\approx 1.6$, one has to expect interferences between the upper and the lower level for the passage of the second crossing, which cannot be resolved by the surface hopping semigroup. The numerical simulations confirm this expectation. 
Figure~\ref{fig:dual}(a) shows that the surface hopping algorithm correctly resolves the first passage, while the non-adiabatic transfer for the second passage is definitely wrong. Since the mean size of the eigenvalue gap and the mean Landau-Zener rate computed by the surface hopping approach qualitatively reflect the true dynamical situation also for the passage of the second crossing, the failure of the approximation must be due to unresolved interlevel interferences. 

\medskip

Also the extended crossing case is not covered by our analysis, since the eigenvalue surfaces do not have a minimal gap but a distance which monotonically decreases as $q\to-\infty$. We have therefore also considered a modified surface hopping algorithm, which allows non-adiabatic transitions at any time step of the numerical simulation if the trajectory's Landau-Zener coefficient is larger than a random number uniformly distributed within the interval $[0,1]$. The outcome of this simulation is marked with circles. The reference dynamics show a monotonously decreasing population of the upper level. The surface hopping algorithm does not initiate any non-adiabatic transfer, since for the trajectories there is no  minimal surface gap. Consequently, it wrongly produces a constant upper level population. The modified surface hopping with unconstrained Landau-Zener transitions starts the non-adiabatic transfer much too early but finally arrives at an upper level population, which is rather close to the true solution.

\appendix
\section{Proof of classical transport}\label{proof:prop}

Here we prove classical transport in the zone of large gap, where $g(q)>R\sqrt\eps$. More precisely, as long as the trajectories of $\Phi^t_\pm$ which reach the support of the observable $a(q,p)$ stay in the region of large gap, we have
\begin{align}\nonumber
&\int_{\R^{2d+1}}\chi(t)
\left(w_\pm^\eps(t)-
w_\pm^\eps(0)\circ\Phi_\pm^{-t}\right)\!(q,p)\,a(q,p)\,\d(q,p,t)\\ \label{eq:ctloc}
&= O(R^{-2}) +  O(R^{-5}\eps^{-1/2})+O(R^{-1}\sqrt\eps).
\end{align}
Note that we shall assume in what follows that $R\sqrt\eps\ll 1$.
The proof of equation~(\ref{eq:ctloc}) relies on symbolic calculus. We shortly recall the two main estimates we are going to use:  

\begin{proposition}\label{prop:symbcal}
For $a,b\in{\mathcal C}_c^\infty(\R^d,\C^{N\times N})$, $N\in\N$, we have 
\begin{equation}\label{est:Schur}
\| \op_\eps(a)\|_{{\mathcal L}(L^2(\R^d))}\leq C \, \sup _{|\beta|\leq d+1} \sup_{x \in\R^d}\int _{\R^d}\left| \partial_\xi^\beta  a(x,\xi)\right|\d \xi.
\end{equation}
for some constant $C>0$ independent of $a$ and $\eps$. Moreover,
\begin{equation}\label{product}
\op_\eps(a)\op_\eps(b)  = \op_\eps(ab)+{\eps\over 2i} \op_\eps(\{a,b\})+\eps^2 R_\eps,
\end{equation}
with $\{a,b\}=\nabla_\xi a \cdot \nabla _x b-\nabla _xa\cdot \nabla_\xi b$ and 
$$\| R_\eps\|_{{\mathcal L}(L^2(\R^d))}\leq C\, \,\sup_{|\alpha|=2, \,2\le |\beta| \le d+3}\,\, \sup_{x\in\R^d} 
\| \partial_x^\alpha \partial_\xi^\beta a(x,\cdot) \|_{L^1(\R^d)}  \| \partial_x^\alpha \partial_\xi^\beta b(x,\cdot) \|_{L^1(\R^d)}$$
for some constant $C>0$ independent of $a$, $b$ and $\eps$.
\end{proposition}

For a proof of Proposition~\ref{prop:symbcal} we refer to \S2 and \S4 of the article on semiclassical pseudodifferential operators in \cite{AFF}, for example. 
Let us now focus on~(\ref{eq:ctloc}) for the plus mode. The proof for the minus mode is analogous. The key argument is an estimate of the commutator 
$$L_\eps = {1\over i\eps} \left[ \op_\eps(c_{\eps,R}\Pi^+), -{\eps^2\over 2} \Delta_q+V(q)\right] ,$$
for  
$$
c_{\eps,R} (q,p)=a(q,p) \,\theta\!\left({g(q)\over R\sqrt\eps} \right)
$$
where $a\in{\mathcal C}_c^\infty(\R^{2d},\R)$ and $\theta\in{\mathcal C}^\infty(\R,\R)$ is a bounded function with bounded derivatives and support outside $0$. We observe that the symbol $c_{\eps,R} \Pi^+$ satisfies the estimate 
\begin{equation*}
\forall \alpha,\beta\in\N^d,\;\;\exists C_{\alpha,\beta}>0,\;\;\forall p,q\in\R^{d},\;\; \left|\partial^\beta_q\partial_p^\alpha (c_{\eps,R}\Pi^+)\right|\leq C_{\alpha,\beta} (R\sqrt\eps)^{-|\beta|},
\end{equation*}
where the constant $C_{\alpha,\beta}$ depends on derivatives of $a$, $\theta$ and $V$. 
We write the potential matrix $V=\lambda^+\Pi^++\lambda^-\Pi^-$ in terms of the eigenvalues and the eigenprojectors and use Proposition~\ref{prop:symbcal} to obtain
$$L_\eps = \op_\eps\left((p\cdot \nabla_q c_{\eps,R}-\nabla\lambda^+\cdot \nabla_p c_{\eps,R})\Pi^+\right) +\op_\eps (B)+O(R^{-2})+O(\sqrt\eps),$$
with 
\begin{eqnarray*}
B &= & c_{\eps,R} \,p\cdot \nabla \Pi^+ -{1\over 2} \left(\lambda^+ (\{c_{\eps,R}\Pi^+,\Pi^+\}-\{\Pi^+,c_{\eps,R}\Pi^+\})\right)\\
& & \qquad -{1\over 2}\left(\lambda^-(\{c_{\eps,R} \Pi^+,\Pi^-\}-\{\Pi^-,c_{\eps,R}\Pi^+\})\right)\\
& = & c_{\eps,R} \,p\cdot \nabla \Pi^+-{g\over 2} \left(\{c_{\eps,R}\Pi^+,\Pi^+\}-\{\Pi^+,c_{\eps,R}\Pi^+\}\right)\\
& = & c_{\eps,R} \,p\cdot \nabla \Pi^+-{g\over 2}\left(\Pi^+ \nabla_p c_{\eps,R} \cdot\nabla \Pi^+
+\nabla_p c_{\eps,R} \cdot \nabla \Pi^+ \Pi^+\right)\\
& = & c_{\eps,R}\, p\cdot \nabla \Pi^+-{g\over 2} \nabla_p c_{\eps,R} \cdot\nabla \Pi^+.
\end{eqnarray*}
where the last equation uses that $\nabla \Pi^+ =\Pi^+\nabla \Pi^++\nabla\Pi^+\Pi^+.$
As a consequence, $B$ is an off-diagonal symbol, in the sense that 
$$B=\Pi^+ B\Pi^- +\Pi^-B\Pi^+.$$
Moreover, we have 
$$
B=B_0+B_1\quad\text{with}\quad B_0= c_{\eps,R} \,p\cdot \nabla \Pi^+,\quad B_1 =  -\tfrac12 g\nabla_p c_{\eps,R} \cdot\nabla \Pi^+.
$$ 
Our next step for proving \eqref{eq:ctloc} is therefore to investigate time averages of off-diagonal observables. 

\medskip
\begin{lemma}~\label{lem:offdiag}
For any $\chi\in{\mathcal C}_c^\infty(\R,\R)$, $j\in\Z$ and any off-diagonal $B_j\in{\mathcal C}_c^\infty(\R^{2d},\C^{2\times 2})$ satisfying the bound
\begin{equation}\label{est:Bj}
\forall \alpha,\beta\in\N^d,\;\;\exists C_{\alpha,\beta}>0,\;\;\forall p,q\in\R^{d},\;\; \left|\partial^\beta_q\partial_p^\alpha B_j\right|\leq C_{\alpha,\beta} (R\sqrt\eps)^{-|\beta|-1+j},
\end{equation}
we have
\begin{equation*}
\int_\R \chi(t) \left\langle\op_\eps(B_j)\psi^\eps_t,\psi^\eps_t\right\rangle\d t =O(\eps ^{1+\frac12(j-3)} R^{-3+j}) \;\;+O(R^{-1}\sqrt\eps).
\end{equation*}
\end{lemma}

Note that in view of and $R\sqrt\eps\ll 1$, we have 
$$\eps ^{1+\frac12(j-3)} R^{-3+j}=\eps (R\sqrt\eps)^{j-2}\gg R^{-1}\sqrt\eps
$$ 
as soon as $j<2$.

\begin{remark}\label{rem:appendix}
The previous Lemma shows in particular that in the large gap region $\{g(q)>R\sqrt\eps\}$, the contribution of the off-diagonal part of the Wigner transform is negligible. Indeed, for all $a\in{\mathcal C}_c^\infty(\R^{2d+2},\C^{2\times 2})$, for all $\chi\in{\mathcal C}_c^\infty(\R)$ and for all  $\theta\in{\mathcal C}^\infty(\R,\R)$ bounded, with bounded derivatives and support outside $0$, we have
$$\int_\R \chi(t) \left\langle\op_\eps\left(\Pi^\pm(q)a(q,p)\Pi^\mp(q) \theta\!\left({g(q)\over R\sqrt\eps}\right)\right) \psi^\eps_t,\psi^\eps_t\right\rangle\d t =O(R^{-2}).$$
This relation comes from the fact that the off-diagonal symbol 
$$
\Pi^\pm(q)a(q,p)\Pi^\mp(q) \theta\!\left({g(q)\over R\sqrt\eps}\right)
$$ 
satisfies Lemma~\ref{lem:offdiag} with $j=1$.
\end{remark}

\begin{proof}
Since $B_j$ is off-diagonal, we can write 
\begin{align*}
B_j  &= [(\Pi^-B_j\Pi^+-\Pi^+B_j\Pi^-) g^{-1} , V_0]\\
&=[(\Pi^-B_j\Pi^+-\Pi^+B_j\Pi^-) g^{-1}, P],
\end{align*}
where we have used that $V_0\Pi^\pm = \pm \frac{1}{2}g\Pi^\pm$.
After quantization, we get 
\begin{align*}
\op_\eps(B_j) =&\left[ \op_\eps((\Pi^-B_j\Pi^+-\Pi^+B_j\Pi^-)  g^{-1}) ,\op_\eps(P)\right]\\
& +O(\eps ^{1+\frac12(j-3)} R^{-3+j}) + O(R^{-1}\sqrt\eps).
\end{align*}
Once applied to $\psi^\eps_t$ which satisfies the Schr\"odinger equation $\op_\eps(P)\psi^\eps_t = 0$, we obtain the announced relation. \end{proof}

Applying Lemma~\ref{lem:offdiag} to $B_1 = -\tfrac12g \nabla_p c_{\eps,R}\cdot\nabla\Pi^+$, we have 
$$
\int_\R \chi(t) \left\langle\op_\eps(B_1)\psi^\eps_t,\psi^\eps_t\right\rangle \d t = O(R^{-2}) + O(R^{-1}\sqrt\eps). 
$$
However, for $B_0 = c_{\eps,R} \,p\cdot\nabla\Pi^+$ the result of Lemma~\ref{lem:offdiag} only provides
$$
\int_\R \chi(t) \left\langle\op_\eps(B_0)\psi^\eps_t,\psi^\eps_t\right\rangle \d t =O(\eps ^{-{1\over 2}} R^{-3}),
$$
an estimate that we want to ameliorate in order to prove \eqref{eq:ctloc}. Therefore, we go one step further in the symbolic calculus and write
\begin{eqnarray*}
\op_\eps(B_0) &=&  \left[ \op_\eps((\Pi^-B_0\Pi^+-\Pi^+B_0\Pi^-)  g^{-1}) , \op_\eps(P)\right]\\
& & + {\eps\over i}\,\op_\eps\left(\left\{ (\Pi^-B_0\Pi^+-\Pi^+B_0\Pi^-) g^{-1} ,\tau +\tfrac12|p|^2 \right\}\right)\\
 & & +{\eps\over 2i} \op_\eps\left((\left\{(\Pi^-B_0\Pi^+-\Pi^+B_0\Pi^-) g^{-1} , V\right\}\right)\\
 &&-\frac{\eps}{2i}\op_\eps\left(\left\{V,(\Pi^-B_0\Pi^+-\Pi^+B_0\Pi^-)  g^{-1}\right\}\right) +O(R^{-4})+O(R^{-1}\eps^{3/2})\\
&=&
 \left[ \op_\eps((\Pi^-B_0\Pi^+-\Pi^+B_0\Pi^-)  g^{-1}) , \op_\eps(P)\right]\\
& &  -{\eps\over i}\, \op_\eps\left(p\cdot \nabla_q \left((\Pi^-B_0\Pi^+-\Pi^+B_0\Pi^-)  g^{-1}\right)\right)\\
& & +O(R^{-2})+O(R^{-1}\sqrt\eps).
\end{eqnarray*}
We set
$$B_{-2}= p\cdot \nabla_q \left((\Pi^-B_0\Pi^+-\Pi^+B_0\Pi^-)  g^{-1}\right),$$
and we observe that $B_{-2}$ satisfies~(\ref{est:Bj}) with $j=-2$. We claim moreover that $B_{-2}$ is off-diagonal, so that,  Lemma~\ref{lem:offdiag} 
gives
$$\int_\R \chi(t) \left\langle\op_\eps(B_{-2})\psi^\eps_t,\psi^\eps_t\right\rangle \d t =O(R^{-5}\eps^{-3/2})  + O(R^{-1}\sqrt\eps)$$
and
$$
\int_\R \chi(t) \left\langle\op_\eps(B_0)\psi^\eps_t,\psi^\eps_t\right\rangle \d t=O(R^{-5}\eps^{-1/2})+O(R^{-1}\sqrt\eps)
$$
It remains to prove the off-diagonal claim. A simple calculation shows 
$$
B_{-2}  =  - g^{-2}(p\cdot\nabla_q g) (\Pi^-B_0\Pi^+-\Pi^+B_0\Pi^-) 
 + g^{-1} p\cdot\nabla_q (\Pi^-B_0\Pi^+-\Pi^+B_0\Pi^-).
$$
Therefore, 
\begin{eqnarray*}
\Pi^\pm B_{-2}\Pi^\pm& =&
g^{-1} \Pi^\pm\big( p\cdot \nabla\Pi^-B_0\Pi^+ +\Pi^-B_0p\cdot \nabla\Pi^+\\
& & \qquad -p\cdot\nabla \Pi^+ B_0\Pi^- -\Pi^+B_0 \,p\cdot \nabla\Pi^-\bigr)\Pi^\pm\\
& =&
2 g^{-1} \Pi^\pm\left(- p\cdot \nabla\Pi^+B_0\Pi^++\Pi^+B_0 \,p\cdot \nabla\Pi^+\right)\Pi^\pm\\
&=& 2 g^{-1} \Pi^\pm\left[B_0\, ,\,p\cdot \nabla \Pi^+\right]\Pi^\pm\\
&=&0
\end{eqnarray*}
since $B_0=c_{\eps,R}\, p\cdot \nabla \Pi^+$.



\subsubsection*{Acknowledgements} 
We thank the anonymous referees for their help in improving the presentation of our results.


\end{document}